\documentclass[sigconf, nonacm]{acmart}

\usepackage[textsize=tiny]{todonotes}
\usepackage{amsfonts, amsmath}
\usepackage{mathtools}
\usepackage{algorithm}
\usepackage{algorithmicx}
\usepackage[noend]{algpseudocode}
\usepackage{enumerate}
\usepackage{enumitem}
\usepackage{xspace}
\usepackage{xcolor}
\usepackage{booktabs} 
\usepackage{tabulary} 
\usepackage[flushleft]{threeparttable} 
\usepackage{subcaption}
\usepackage{environ}

\usepackage{letltxmacro}
\LetLtxMacro{\oldsqrt}{\sqrt}
\renewcommand{\sqrt}[2][\mkern8mu]{\mkern-4mu\mathop{\oldsqrt[#1]{#2}}}

\newtheorem{thm}{Theorem}[section]
\newtheorem{defn}[thm]{Definition}
\newtheorem{remark}[thm]{Remark}
\newtheorem{fact}[thm]{Fact}
\newtheorem{lem}[thm]{Lemma}
\newtheorem{cor}[thm]{Corollary}

\newif\ifcomment
\commentfalse

\newcommand{\phil}[1]{\ifcomment\todo[color=red!50]{Phil: #1}\fi}

\newif\iflong
\longtrue

\iflong\commentfalse\fi

\iflong
\renewcommand\footnotetextcopyrightpermission[1]{} 
\fi

\newcommand{\lng}[1]{\iflong#1\fi}
\newcommand{\shrt}[1]{\iflong\else#1\fi}

\NewEnviron{longversion}{\iflong\expandafter\BODY\fi}
\NewEnviron{shortversion}{\iflong\else\expandafter\BODY\fi}

\newcommand{\LOCAL}{\ensuremath{\mathsf{LOCAL}}\xspace}
\newcommand{\CONGEST}{\ensuremath{\mathsf{CONGEST}}\xspace}
\newcommand{\NCC}{\ensuremath{\mathsf{NCC}}\xspace}
\newcommand{\HYBRID}{\ensuremath{\mathsf{HYBRID}}\xspace}
\newcommand{\CC}{\ensuremath{\mathsf{CLIQUE}}\xspace}
\newcommand{\BCC}{\ensuremath{\mathsf{BCC}}\xspace}

\newcommand{\bigO}{\smash{\ensuremath{O}}}
\newcommand{\tilO}{\smash{\ensuremath{\widetilde{O}}}}
\newcommand{\tilOm}{\smash{\ensuremath{\widetilde{\Omega}}}}
\newcommand{\tilT}{\smash{\ensuremath{\widetilde{\Theta}}}}
\newcommand{\tild}{\smash{\ensuremath{\widetilde{d}}}}
\newcommand{\tilD}{\smash{\ensuremath{\widetilde{D}}}}

\newcommand{\eps}{\varepsilon}
\newcommand{\calA}{\mathcal{A}}
\newcommand{\calB}{\mathcal{B}}

\newcommand{\calS}{\mathcal{S}}
\newcommand{\E}{\mathbb{E}}
\newcommand{\p}{\!+\!}
\newcommand{\m}{\!-\!}
\renewcommand{\Pr}{\mathbb{P}}

\newcommand{\set}[1]{\ensuremath{\left\{#1\right\}}}

\DeclareMathOperator{\polylog}{polylog}

\DeclareMathOperator*{\argmin}{arg\,min}

\newcommand{\hide}[1]{}

\AtBeginDocument{%
	\providecommand\BibTeX{{%
			\normalfont B\kern-0.5em{\scshape i\kern-0.25em b}\kern-0.8em\TeX}}}

\copyrightyear{2020}
\acmYear{2020}
\setcopyright{acmlicensed}\acmConference[PODC '20]{ACM Symposium on Principles of Distributed Computing}{August 3--7, 2020}{Virtual Event, Italy}
\acmBooktitle{ACM Symposium on Principles of Distributed Computing (PODC '20), August 3--7, 2020, Virtual Event, Italy}
\acmPrice{15.00}
\acmDOI{10.1145/3382734.3405719}
\acmISBN{978-1-4503-7582-5/20/08}

\begin{document}
	
	\title{Computing Shortest Paths and Diameter in the Hybrid Network Model}
	

	\author{Fabian Kuhn}
	\email{kuhn@cs.uni-freiburg.de}
	\affiliation{%
		\institution{University of Freiburg}
	}
	\author{Philipp Schneider}
	\email{philipp.schneider@cs.uni-freiburg.de}
	\affiliation{%
		\institution{University of Freiburg}
	}
	
	\renewcommand{\shortauthors}{}
	
	\begin{abstract}
		The $\mathsf{HYBRID}$ model, introduced in [Augustine et al., SODA '20], provides a theoretical foundation for networks that allow multiple communication modes. The model follows the principles of synchronous message passing, whereas nodes are allowed to use \textit{two} fundamentally different communication modes. First, a local mode where nodes may exchange arbitrary information per round over edges of a local communication graph $G$ (akin to the $\mathsf{LOCAL}$ model). Second, a global mode where every node may exchange $O(\log n)$ messages of size $O(\log n)$ bits per round with arbitrary nodes in the network. The $\mathsf{HYBRID}$ model intends to reflect the conditions of many real hybrid networks, where high-bandwidth but inherently local communication is combined with highly flexible global communication with restricted bandwidth.

We continue to explore the power and limitations of the $\mathsf{HYBRID}$ model by investigating the complexity of computing shortest paths and diameter of the local communication graph $G$. We show that the all pair shortest paths problem can be solved exactly in \smash{$\widetilde O\big(\sqrt n\big)$} rounds, which improves on the previous \smash{$\widetilde O\big(n^{2/3}\big)$} round algorithm and closes the gap to the known $\widetilde \Omega\big(\sqrt n\big)$ lower bound (up to $\polylog n$ factors). 
Furthermore, we give constant approximations for the $k$-source shortest paths problem ($k$-SSP) with runtime \smash{$\widetilde O\big(\sqrt{k}\big)$}, provided that $k$ is sufficiently large.
As $k$-SSP has a lower bound of \smash{$\widetilde\Omega\big(\sqrt{k}\big)$} even for large approximation ratios, 
our $k$-SSP algorithms are almost tight for large enough $k$. In the case of a single source we give an exact \smash{$\widetilde O\big(n^{2/5}\big)$}-round algorithm, improving on the known \smash{$\widetilde O\big(\sqrt{D}\big)$}-round algorithm for graphs with large diameter $D$.

For the diameter problem we provide algorithms with complexities  \smash{$\widetilde O\big(n^{1/3}/\varepsilon\big)$} and \smash{$\widetilde O\big(n^{0.397}/\varepsilon\big)$} and approximation factors $(3/2\!+\!\varepsilon)$ and $(1\!+\!\varepsilon)$, respectively. On the negative side, we demonstrate that the classical 2-party set-disjointness framework can be adapted for the $\mathsf{HYBRID}$ model to prove a lower bound of $\widetilde\Omega\big(n^{1/3}\big)$ rounds for computing the diameter exactly. For the \textit{weighted} diameter problem the same holds for computing $(2 \!-\!\varepsilon)$-approximations for any $\varepsilon > 0$.

	
	\end{abstract}

	\maketitle
	
	\section{Introduction}

In networking contexts the communication nodes often have access to different modes of communication with fundamentally different properties \cite{gmyr17_hybrid,Augustine2020}.
For instance, in a network of mobile wireless devices the nodes are able to communicate through the cellular networking infrastructure, however close devices can also communicate locally with their wireless interfaces \cite{KS18}. While point-to-point communication between any two devices through the cellular network is more flexible and faster for devices that are at great distance, short ranged communication among close devices is typically cheaper and tends to allow higher bandwidth. In a different scenario, a company or organization might combine their own local area network with communication over the Internet \cite{RS11,TBKC18}. For data centers there are suggestions to augment the standard wired electrical communication infrastructure with high-speed optical or wireless connections \cite{FPRBSFPV10,HKPBW11}. Solving shortest paths and diameter problems in the local communication infrastructure of such \emph{hybrid networks} has direct applications, e.g., for learning the topology of the local network which can be used for efficient IP-routing.

\paragraph{The Hybrid Network Model.} From an algorithmic point of view, it is interesting to ask whether the combination of multiple communication modes can lead to an asymptotic advantage over using one of the communication modes by itself. The recent paper by Augustine et al.~\cite{Augustine2020} has introduced the \emph{hybrid network model} (in the following called \HYBRID) as a formal distributed communication model to study this question. The basic assumption underlying the \HYBRID model is that the nodes of a network can use two modes of communication: (1) a relatively low bandwidth but flexible, global communication mode, where any two nodes can directly communicate with each other and (2) a high bandwidth local communication mode, where only neighboring nodes in the local network can directly communicate among each other.

Formally, the network consists of a set $V$ of $n$ nodes with IDs $[n]$\footnote{We define $[k] := \{1, \ldots , k\}$ for $k \in \mathbb{N}$.} and a local communication topology given by a graph $G=(V,E)$. Communication happens in synchronous rounds. In each round nodes can do arbitrary computations based on the information they posses.
Local communication is modeled with the standard \LOCAL model~\cite{peleg00,linial92}: in each round, every node can exchange an arbitrary message with each of its neighbors in $G$.
The global communication mode uses the recently introduced \emph{node-capacitated clique} (\NCC) model~\cite{augustine19_NCC}: in each round, every node can exchange $O(\log n)$-bit messages with up to $O(\log n)$ arbitrary nodes.

As discussed in \cite{Augustine2020}, one can parameterize hybrid networks by the maximum message size $\lambda$ for the local mode and the number of bits $\gamma$ each node can exchange using the global mode. Then the standard \LOCAL and \CONGEST, the \CC and the \NCC model are just marginal cases of that parametrization.\footnote{\LOCAL: $\lambda = \infty, \gamma = 0$, \CONGEST: $\lambda = \bigO(\log n), \gamma = 0$, \CC (+``Lenzen-Routing'' \cite{Lenzen13}): $\lambda = 0, \gamma = n \log n$, \NCC: $\lambda = 0, \gamma = \bigO(\log^2 n)$.}
Considering the whole spectrum of \textit{non-marginal} hybrid network parameterizations, the combination \LOCAL + \NCC pushes both communication modes to one extreme end of the spectrum, as it leaves local communication completely unrestricted and allows only $\polylog n$ bits of global communication per node.
Note that restricting the global communication regime even further, e.g., to just $\bigO(\log n)$ bits per node and round would not change our main narrative. Any \HYBRID model algorithm could still be simulated in such a slightly more restrictive model at the cost of a $\polylog n$ multiplicative factor. Permitting each node to exchange $\bigO(\log n)$ many messages of size $\bigO(\log n)$ bits per round simplifies the formulation of randomized algorithms. Except for lower bounds, we neglect $\polylog n$ factors anyway by using the $\tilO(\cdot)$-notation.\footnote{The $\tilde{O}(\cdot)$ and $\tilde{\Omega}(\cdot)$ notations hide factors that are polylogarithmic in $n$.}

\paragraph{Distance Computations in the \HYBRID Model.} As in \cite{Augustine2020}, we investigate the complexity of shortest paths problems in the local communication graph $G$ and of (approximately) computing the diameter of $G$ for the \HYBRID communication model. When only using the \LOCAL model, computing the diameter or shortest paths for any number of sources clearly takes $\Theta(D)$ rounds, where $D$ is the diameter of $G$. This is also true if distances only need to be computed approximately (for reasonably good approximations). Therefore, all of our algorithmic results can be interpreted as the minimum of $D$ and the given round complexity. Note that there are graphs, for which $D$ is linear in $n$.

If only the \NCC model is used to communicate, the (approximate) APSP problem clearly requires $\tilOm(n)$ rounds because in each round, each node can only receive $\bigO(\log^2 n)$ bits of information and learning the distances to all other nodes requires at least $\Omega(n)$ bits of information. Even for the (approximate) SSSP problem nothing better than an $\tilO(n)$ bound is known for general graphs in the \NCC model. 
As the work of \cite{Augustine2020} demonstrates, in the \HYBRID model, i.e., the combination of \LOCAL and \NCC, significantly faster algorithms are possible than in each individual model. 


\subsection{Contributions}

The key algorithmic tool used in the APSP algorithms of \cite{Augustine2020} is an algorithm to solve a multi-message broadcast problem (called the \textit{token dissemination problem}) in the \HYBRID model. In this work, we consider a more general problem, which we call the \textit{token routing problem}. We are given a set of sources $S \subseteq V$ from which we must route a total of $K$ point-to-point messages of size $O(\log n)$ bits each, where each node is source and destination of at most $k$ messages.
We show that if the sources and destinations are sufficiently well spaced in the local network, then all messages can be routed in time \smash{$\tilO\big(K/n + \sqrt{k} + \sqrt{|S|}\big)$}. By contrast, simply broadcasting all messages using the algorithm of \cite{Augustine2020} takes \smash{$\tilOm\big(\sqrt{k|S|}\big)$}. We refer to Section \ref{sec:token-routing} for more details.

Leveraging the token routing protocol we obtain our first result in Section \ref{sec:apsp}, where we speed up a key step in the  exact APSP algorithm of \cite{Augustine2020} that constituted a runtime bottleneck in the original algorithm. With this modification we improve on the previous runtime of $\tilO(n^{2/3})$ and close the gap to the known lower bound $\tilOm\big(\sqrt{n}\big)$ (up to $\polylog(n)$ factors).

\begin{thm}\label{thm:APSP}
	There is an algorithm that solves APSP exactly in the \HYBRID model in $\tilO\big(\sqrt{n}\big)$ rounds w.h.p.
\end{thm}

 As another direct application of the token routing algorithm, we show in Section \ref{sec:cc-sp-sim} that a single round of the much more powerful \CC model\footnote{The congested clique model (or \CC model) refers to the synchronous message passing model, where communication among nodes works as follows: in each round each node may send a $\bigO(\log n)$ bit message to every other node.} can be simulated on (a suitably chosen) subset of $n^x$ ($x \leq 1$) nodes of $G$ in  $\tilO(n^{x/2} \p n^{2x-1})$ rounds in the \HYBRID model. We first apply this result to the $k$-source shortest paths problem ($k$-SSP) problem, where each node needs to compute (approximate) distances to an arbitrary subset of $k$ nodes of $G$. 
 
We provide a framework algorithm for $k$-SSP that intuitively works as follows. It takes a \CC-algorithm $\calA$ that approximates $k$-SSP as input and uses it to solve $k$-SSP on a small skeleton-graph $\calS$ of $G$.\footnote{A skeleton graph $\calS$ is a graph on a random subset of nodes of $G$ that conserves the distances of $G$. See Appendix \ref{apx:skeleton-graphs} for details and properties.} The solution on $\calS$ is then extended to the whole graph $G$ using the local network.
Combining this framework algorithm with the fast \CC algorithms of \cite{Censor-Hillel2019a,Censor-Hillel2019b}, we obtain an array of results (with potential for further improvement if faster or more precise algorithms for the \CC model are found).

 
 \begin{thm}\label{thm:kSSP}
 	The $k$-SSP problem can be solved in the \HYBRID model w.h.p.\ for $\eps > 0$ and the following parameters.
 	\begin{itemize}[leftmargin=19pt]
 		\item In \smash{$\tilO\big({n^{1/3}}/{\eps}\big)$} rounds for \smash{$k = n^{1/3}\!$} sources and approximation factor $(3\!+\!\eps)$ on  weighted and $(1\!+\!\eps)$ on unweighted graphs.
 		\item In \smash{$\tilO\big({n^{1/3}}/{\eps} \!+ \sqrt{k}\big)$} rounds and approximation factor $(7\!+\!\eps)$ on weighted and $(2\!+\!\eps)$ on unweighted graphs.
 		\item In \smash{$\tilO\big(n^{0.397} \!+ \sqrt{k}\!\big)$} rounds and approximation factor $(3\!+\!o(1))$ on weighted graphs and runtime \smash{$\tilO\big(n^{0.397}/\eps + \sqrt{k}\!\big)$} and approximation factor $(1\!+\!\eps)$ on unweighted graphs.
 	\end{itemize}
 \end{thm} 
 
Furthermore, we apply the above framework algorithm to the fast exact SSSP algorithm of \cite{Censor-Hillel2019a}, which yields an $\tilO(n^{2/5})$-round, exact SSSP algorithm for the \HYBRID model. Note that this improves on the $\tilO\big(\sqrt{SPD}\big)$ shortest path algorithm of \cite{Augustine2020} for graphs with large shortest path diameter $SPD$.

\begin{thm}\label{thm:SSSP}
	There is an algorithm that solves SSSP exactly in the \HYBRID model in $\tilO(n^{2/5})$ rounds w.h.p.
\end{thm}

Apart from studying shortest paths problems, we also consider a more basic distance problem, namely the problem of computing the diameter of the given network graph $G$. In Section \ref{sec:cc-diam-sim}, we reuse the observation that we can simulate the \CC model on small skeleton graphs of $G$, in order to simulate fast diameter algorithms from \cite{Censor-Hillel2019a,Censor-Hillel2019b} for the \CC model. We show that the diameter of the skeleton graph yields decent diameter approximations for the whole graph $G$. We obtain the following upper bounds.

\begin{thm}\label{thm:diameter}
 	The diameter can be computed w.h.p.\ in the \HYBRID model for $\eps > 0$ and the following parameters. In \smash{$\tilO\big(n^{1/3}/\eps\big)$} rounds with approximation factor $(3/2\!+\!\eps)$ and in \smash{$\tilO\big(n^{0.397}/\eps\big)$} rounds with approximation factor $(1\!+\!\eps)$.
\end{thm}

In \shrt{Section \ref{sec:lower-bounds}}\lng{Sections \ref{sec:k-src-lb} and \ref{sec:diam-lb}} we provide algorithmic lower bounds for the \HYBRID model. First we slightly modify the \smash{$\tilOm\big(\sqrt{n}\big)$} lower bound for the APSP problem given in \cite{Augustine2020} to obtain a \smash{$\tilOm\big(\sqrt k\big)$} lower bound for the $k$-SSP problem even for large approximation factors. This implies that our corresponding upper bounds are almost tight if the number of sources is sufficiently large ($k \geq n^{2/3}$ suffices).

\begin{thm}
	\label{thm:k-src-lb}
	A randomized algorithm that computes an {$\alpha$} approximation of the $k$-source shortest paths problem in the \HYBRID model takes \smash{$\tilOm\big(\sqrt k\big)$} rounds, even on unweighted graphs, for any $\alpha \leq \alpha'$ and some fixed \smash{$\alpha' \in \Theta\big(n/\,\sqrt{k}\big)$}.
\end{thm}

Furthermore, we show that the $2$-party set disjointness framework can be adapted for the \HYBRID model. Specifically, we adapt the worst case graphs of \cite{FHW12, holzer_et_al16} (constructed for the \CONGEST model) to create an according information bottleneck in the \HYBRID model. We obtain the following lower bounds.

\begin{thm}\label{thm:diameter_lower}
	Computing the diameter in the \HYBRID model takes $\Omega\big((n/\log^{2} n)^{1/3}\big)$ rounds even for randomized algorithms. For \emph{weighted} diameter the same holds for $(2\m \eps)$-approximations for any $\eps >0$ (weights up to $n^{1/3}$ suffice). 
\end{thm}

\lng{Note that for the \textit{weighted} diameter problem there is an upper bound of $\tilO(n^{1/3})$ rounds for a $(2 \p o(1))$-approximation. This upper bound is implied by the $(1 \p o(1))$-approximative \HYBRID-algorithm for SSSP with same complexity given in \cite{Augustine2020}.\footnote{Due to the fact that the eccentricity $e(v) := \max_{u \in V} d(u,v)$ of a node can be computed with a SSSP algorithm and $D/2 \leq e(v) \leq D$.} 
}

\subsection{Related Work}

\paragraph{Hybrid Communication Networks} Apart from the work of Augustine et al.~\cite{Augustine2020}, there is little work on distributed algorithms in hybrid network models. In \cite{Augustine2020} an information dissemination routine to broadcast $k$ tokens in \smash{$\tilO\big(\sqrt{k}\big)$} time is employed to solve various shortest paths problems. They solve APSP exactly in time \smash{$\tilO\big(n^{2/3}\big)$} and show that approximations have complexity \smash{$\tilT\big(\sqrt{n}\big)$}. For SSSP they give an exact solution in time $\tilO\big(\sqrt{SPD}\big)$ for graphs with shortest path diameter $SPD$. In a similar approach to this paper, the token dissemination protocol is used to simulate the \textit{broadcast congested clique} (\BCC) model. This results in a $(1 \p o(1))$-approximation with complexity \smash{$\tilO(n^{1/3})$} when combined with the fast SSSP algorithm by \cite{BKKL17} for that model. Finally, they give an approximation scheme for SSSP with constant but large factor $(1/\eps)^{O(1/\eps)}$ in $\tilO(n^\eps)$ rounds.

\paragraph{Global Communication Networks} Notably, the work of Gmyr et al.~\cite{gmyr17_hybrid} studies the minimum spanning tree (MST) problem in a similar setting as \cite{Augustine2020}. However, their algorithm uses only the global communication edges (with an explicit polylogarithmic degree overlay network). They show that if the local network has bounded degree, an MST can be computed in polylogarithmic time. Following up on \cite{gmyr17_hybrid}, Augustine et al.~\cite{augustine19_NCC} introduce the more abstract \NCC model and show that the MST problem can be solved in polylogarithmic time on general graphs. The \NCC model can be seen as a (much weaker) generalization of the \textit{congested clique} (\CC) model (however in general, there is an $\Omega(n/\log n)$-factor slowdown when going from the \CC to the \NCC model). Shortest paths problems have been intensively studied in the \CC model (see, e.g., \cite{censor-hillel18,Censor-Hillel2019a,Censor-Hillel2019b,holzer_et_al16,legall16}). Notably, \cite{Censor-Hillel2019b} gives polylogarithmic-time, constant approximation algorithms for APSP and $k$-SSP for the \CC model. As discussed above, many of our results are obtained by simulating existing \CC algorithms in the \HYBRID model on small, random subgraphs.

\paragraph{Local Communication Networks} In the \CONGEST model, there is a long series of work on exactly and approximately solving the SSSP problem (see, e.g., \cite{Sarma2012,LP13,Nan14,HKN16,BKKL17}) and the APSP problem (see, e.g., \cite{LeP13,HW12,LP15,Nan14,HNS17,ARKP18,BN19}). For SSSP, we know how to compute an $(1 \p o(1))$-approximation in time $\tilO(D\!+\sqrt{n})$ due to~\cite{BKKL17}, which is known to be optimal up to polylogarithmic factors even for a much coarser approximation~\cite{Sarma2012}. For the APSP problem, it was shown that $\tilOm(n)$ rounds are needed, even for a polynomial approximation by~\cite{Nan14}. Recently, it was shown that this bound can be matched (up to polylogarithmic factors) even for computing APSP exactly~\cite{BN19}.

\paragraph{Graph Diameter} Work on distributed algorithms for the diameter of a graph was initiated by Frischknecht et al.~\cite{FHW12}. They show that in the \CONGEST model, one needs $\tilOm(n)$ rounds for an exact solution and $\tilOm(D+\sqrt{n})$ rounds for a $(\frac{3}{2} \m \eps)$-approximation even in unweighted graphs. Both bounds have later been achieved: in \cite{PRT12}, an elegant $O(n)$-time algorithm for unweighted graphs was given and in \cite{holzer14_BA}, an $\tilO(D+\sqrt{n})$-time algorithm for computing a $\tfrac{3}{2}$-approximation of the diameter in unweighted graphs was presented.
For weighted graphs, computing a $(2-\eps)$-approximation takes $\tilOm(n)$ rounds~\cite{holzer_et_al16}. The diameter problem has also been considered in the \CC model by \cite{Censor-Hillel2019b}. They show that a $(\frac{3}{2} \p \eps)$-approximation can be computed in polylogarithmic time.

\subsection{Preliminaries}

\paragraph{General Definitions} The scope of this paper is solving graph problems in the local communication graph $G = (V,E)$. We consider edges to be \textit{undirected}. Edges are assigned \textit{weights} $w: E \to [W]$, where $W$ is at most polynomial in $n$, thus the weight of an edge (and a simple path) fits into a message. A graph is considered \textit{unweighted} if $W=1$. Let $w(P) = \sum_{e \in P}w(e)$ denote the length of a path $P \subseteq E$. Then the \emph{distance} between two nodes $u,v \in V$ is defined as
\[
	d(u,v) := \!\min_{\text{$u$-$v$-path } P} w(P).
\] 
A path with smallest length between two nodes is called a \emph{shortest path}.
Let $|P|$ be the number of edges (or \emph{hops}) of a path $P$.
The \emph{hop-distance} between two nodes $u$ and $v$ is defined as 
\[
	hop(u,v) := \!\min_{\text{$u$-$v$-path } P} |P|.
	\vspace*{-1mm}
\] 
The \emph{diameter} of $G$ is defined as 
\[
	D(G) := \max_{u,v \in V} hop_G(u,v).
\]
Let the \emph{$h$-limited distance} from $u$ to $v$ 
\[
	d_{h}(u,v) := \!\!\min_{{\text{$u$-$v$-path } P, |P| \leq h }}\, w(P).
\]
If there is no $u$-$v$ path $P$ with $|P|\leq h$ we define $d_{h}(u,v) := \infty$. 

\paragraph{Problem Definitions}

In the general $k$-\textit{sources shortest paths problem} ($k$-SSP) we are a given set of $k$ sources. Every $v \in V$ has to learn $d(v,s)$ for all sources $s$. In the $(\alpha,\beta)$-approximate version of the problem for $\alpha \geq 1, \beta \geq 0$, every $u\in V$ has to learn values $\tilde{d}(v,s)$ such that $d(v,s)\leq \tilde{d}(v,s)\leq \alpha\cdot d(v,s) + \beta$ for all sources $s$. For $\beta = 0$ we speak of the ``$\alpha$-approximation'' problem.

The \textit{all-pairs shortest paths problem} (APSP) equals the case $k=n$ and the single-source shortest paths problem (SSSP) equals the case $k=1$. In the \textit{diameter problem} \textit{all} nodes have to learn $D(G)$. We define $(\alpha,\beta)$-approximations of the diameter problem analogously. 

The \textit{token routing problem} describes the following information routing problem. We are given a set of sender nodes $S \subseteq V$ and a set of receiver nodes $R \subseteq V$. Each sender $s \in S$ must send at most $k_S$ and each receiver $r \in R$ must receive at most $k_R$ tokens (of size $O(\log n)$ bits each), respectively. Each token has a dedicated receiver node $r \in R$, and each receiver $r \in R$ knows the senders it must receive a token from and how many tokens it needs to receive from each sender. The problem is solved when all nodes in $R$ know all tokens they are the receiver of.

	\section{Token Routing}
\label{sec:token-routing}

In the first technical section we solve the token routing problem defined above if sender nodes $S$ and receiver nodes $R$ are ``well distributed'' which is the case (w.h.p.) if senders and receivers are sampled from $V$ uniformly at random. We start by defining \textit{helper sets} for the sender nodes $S$ and receiver nodes $R$, respectively:

\begin{defn}
	\label{def:helpers}
	Let $W \subseteq V$ and let $k \in \bigO\big(n/|W|\big)$ be the maximum number of tokens per node. A family $\{H_w \subseteq V \mid w \in W\}$ of \emph{helper sets} fulfills the following properties for all $w \in W$ and some integer $\mu \in \smash{\Theta \big(\!\min\big(\sqrt{k}, n/|W|\big)\big)}$. (1) Each $H_w$ has size at least $\mu$. (2) The hop-distance is $hop(w,x) \in \tilO(\mu)$ between $w$ and any $x \in H_w$. (3) Each node is member of at most $\tilO(1)$ sets $H_w$. 
\end{defn}

Figuratively speaking, a family of helper sets assigns each node in $W$ a sufficiently large set of helper nodes nearby, so that no node from $V$ has to help too many nodes in $W$. Suppose the sets $H_s$ for $s \in S$ and $H'_r$ for $r\in R$ form two families of helper sets for senders $S$ and receivers $R$, respectively. Intuitively, each sender $s \in S$ (or receiver $r \in R$) can increase its bandwidth over the global network by a factor of $\mu$, by relying on its $\mu$ helpers from $H_s$ (or $H'_r$).

This works as follows \big(assume for the sake of this explanation that $|S| \in \Theta\big(|R|\big)$ and $k := k_S \in \Theta\big(k_R\big)$, so that $\mu := \mu_S \in \Theta(\mu_R)$\big). In a first step we use the local network to distribute the tokens of each sender to its respective helpers. Due to property (2) this can be done in $\tilO(\mu)$ rounds (each sender simply floods its tokens to depth $\tilO(\mu)$ in the local network and each helper picks the tokens assigned to it by its sender).

Then each helper $v \in H_s$ takes care of sending some part of the tokens of $s \in S$ to some helper $u \in H'_r$ of the according receiver $r \in R$ of that message. We ensure that this happens in a balanced manner, so that each helper has to exchange at most $\tilO(k/\mu)$ tokens, due to property (1). The main challenge in this step is to pair up helpers of senders with helpers of receivers in a balanced way without helper sets knowing each other. We achieve this by relaying tokens via pseudo-random intermediate nodes (i.e., determined by a publicly known, randomly seeded hash-function).

More precisely each helper $v \in H_s$ forwards the messages it is responsible for, to a set of intermediate nodes. Afterwards each helper $u \in H'_r$ requests its share of all messages to $r$ from the appropriate intermediate nodes. Here we leverage that a receiver $r \in R$ knows the set of messages it must receive from each sender, so that its helpers $u \in H'_r$ can compute which intermediates they must fetch their share of messages from.



After all tokens are sent in that manner, the helper sets $H'_r$ will collectively own all tokens for receiver $r \in R$. Finally, $r$ collects its tokens from $H'_r$ in time $\tilO(\mu)$ via the local network (again the helpers can simply flood their tokens to depth $\mu$ and the receivers pick the tokens they are recipient of). The whole process takes \smash{$\tilO\big(\mu + k/\mu \big) = \tilO\big(K/n + \sqrt{k}\big)$} rounds, where \smash{$K \in \Theta \big(|S|\!\cdot\!k\big)$}.

Note that the helper sets specified in Definition \ref{def:helpers} can only exist if the sets of senders $S$ and receivers $R$ are sufficiently separated from each other in the network. We will see that we can find appropriate helper sets if $S \subseteq V$ and $R \subseteq V$ are sampled uniformly at random, which is sufficient for our purposes. Ultimately, we will show the following theorem.

\begin{thm}
	\label{thm:token-routing}
	Let $S, R \subseteq V$ be sampled from $V$ with probabilities \smash{$p_S = 1/n^{\eps}, p_R = 1/n^{\delta}$} for const.\ $\eps, \delta \in [0,1)$, respectively. Let $k_S$ and $k_R$ be the number of tokens to be sent or received by any node in $S$ and $R$, respectively.
	Let $K := |S|\!\cdot\!k_S \!+\! |R|\!\cdot\!k_R$ be the total workload. Then the token routing problem can be solved in \smash{$\tilO\big(K/n +\, \sqrt{k_S} +\, \sqrt{k_R}\big)$} rounds.
\end{thm}

\subsection{Computing Helper Sets}
\label{sec:helper-sets}

Our first task is to compute helper sets satisfying the conditions of Definition \ref{def:helpers}.
\shrt{As this task is not particularly demanding we restrict ourselves to a high level summary of our approach, and skip most of the details. The detailed subroutines a complete proofs are provided in the full version of this article \cite{KS20}.}
We propose a simple protocol that leverages fast symmetry breaking algorithms that are known for the local network (even for the more restrictive $\mathsf{CONGEST}$ model). Specifically, we first compute a $(\alpha, \beta)$-ruling set. Such a set has the following properties:

\begin{defn}[c.f., \cite{awerbuch89}]
	An $(\alpha, \beta)$-ruling set for $G = (V,E)$ is a subset $\mathcal{R} \subseteq V$, such that for every $v \in V$ there is a node $r \in \mathcal{R}$ with $hop(v,r) \leq \beta$ and for any $r_1,r_2 \in \mathcal{R}, r_1 \neq r_2$ we have $hop(r_1,r_2) \geq \alpha$.
\end{defn}

There is a plethora of work on how to compute ruling sets in the distributed setting. A particular result that suits our needs in terms of parameter range and runtime is given by the following Lemma \ref{lem:ruling_set}. This lemma is a Corollary of Theorem 1.1 in \cite{kuhn2018} (which also works in the stronger \CONGEST model). The classic paper \cite{awerbuch89} provides the same result for \LOCAL, which would also be sufficient.

\begin{lem}[c.f., \cite{awerbuch89,kuhn2018}] 
	\label{lem:ruling_set}
	A $\big(2\mu\!+\!1, 2\mu \lceil \log n \rceil\big)$-ruling set can be computed (deterministically) in the local network in $\bigO(\mu \log n)$ rounds.
\end{lem}

\lng{Before we start with the details, we give some high level view how Algorithm \ref{alg:compute-helpers} (given below) uses a ruling set to compute helper sets.} Based on an $(\alpha, \beta)$-ruling set $\mathcal{R} \subseteq V$, we can compute a \textit{clustering} of $G$. Initially, each ``ruler'' $r \in \mathcal{R}$ represents a separate cluster $C_r$. Then each node $v \in V$ determines the ruler $r \in \mathcal{R}$ closest to it (break ties arbitrarily), and joins its cluster $C_r$. 
Now we have the following properties. Each cluster contains at least $(\alpha\!-\!1)/2$ nodes (due to the min.\ hop-distance of $\alpha$ between rulers). Furthermore, the diameter of each cluster is at most $\beta$. 

Presume that each node joins some set $W$ with probability $p \in (0,1)$. That means, $W$ is randomly sampled set of nodes that persists in the network in a distributed sense, i.e., each node is aware whether it is a member of $W$ or not. The goal is to compute helper sets for $W$ based on the aforementioned clustering.
First we compute such a clustering based on a $\big(2\mu\!+\!1,2\mu\lceil \log n \rceil\big)$-ruling set as described above. Note that $\mu$ is the integer parameter used in Definition \ref{def:helpers} (but neglect its exact value for now). Let $C$ be one of the clusters. The number of nodes from $W$ within that cluster is roughly $|C \cap W| \in \tilO(p \!\cdot\! |C|)$ w.h.p. Then each node $v$ in cluster $C$ joins the helper set $H_w$ of each $w \in W \cap C$ with probability $q \in \tilT(\mu/|C|)$. 

Now we have the following properties: (1) each $H_w$ has size $|H_w| \in \tilT\big(q \, |C|\big) = \tilT(\mu)$ w.h.p. In fact, we can adjust the constant factors in $q$ such that $|H_w|\geq \mu$ w.h.p. (2) the hop-distance between any pair of nodes in the same cluster $C$ is at most $2\mu\lceil \log n \rceil \in \tilO(\mu)$ and therefore the same is true for any helper set $H_w \subseteq C$. (3) each $v \in C$ joins roughly $\tilO\big(q |C \cap W|\big) = \tilO\big(qp |C|\big) = \tilO(\mu p) = \tilO(\mu|W|/n)$ helper sets. Then for $\mu \!\in\! \tilO(n/|W|)$ (as required by Definition \ref{def:helpers}) we have that each $v$ joins at most $\tilO(1)$ sets $H_w$. Note that these were the three requirements for helper sets of Definition \ref{def:helpers}.

\begin{longversion}
\begin{algorithm}[H]
	\caption{\texttt{Compute-Helpers}$(\mu)$ \Comment{\textit{$\mu$ as in Lemma \ref{lem:helpers}}}}
	\label{alg:compute-helpers}
	\begin{algorithmic}
		\State compute a $\big(2\mu\!+\!1,2\mu\lceil \log n \rceil\big)$-ruling set $\mathcal{R}$ with algorithm \cite{kuhn2018}
		\For {$2\mu\lceil \log n \rceil$ rounds} \Comment{\textit{learn $\mathcal{R}$, $W$ within $\mu\lceil \log n \rceil$ hops}}
		\State $v$ forwards information on $\mathcal{R},W$ via incident local edges
		\EndFor
		\State $\mathcal{R}_v \gets$ rulers that $v$ discovered
		\State $r \gets \argmin_{r' \in \mathcal R_v} hop(v,r')$ \Comment{\textit{closest ruler}}
		\State $v$ joins cluster $C_r$ \Comment{\textit{establish clustering}}
		\For {$4\mu\lceil \log n \rceil$ rounds} \Comment{\textit{$v$ learns all members of $C_r$}}
		\State $v$ forwards information on $C_w$ with $w \in \mathcal{R}$ via local edges
		\EndFor
		\State $q \gets \min(2\mu/|C_r|, 1)$ 
		\For {each $w \in W \cap C_r$}
			\State $v$ joins $H_w$ with probability $q$ 
		\EndFor
	\end{algorithmic}
\end{algorithm}
\begin{lem}
	\label{lem:helpers}
	Let $W \subseteq V$ be a random set obtained with sampling probability \smash{$p = \tfrac{1}{n^{\eps}}$} for constant $\eps \in (0,1)$. Let $k \in \mathbb N$. Running Algorithm \ref{alg:compute-helpers} with \smash{$\mu := \big\lfloor\!\min\big(\sqrt{k}, 1/p\big) \big\rfloor$} establishes a family of helper sets $\{H_w \subseteq V \mid w \in W \}$ with respect to Definition \ref{def:helpers} w.h.p. 
	Algorithm \ref{alg:compute-helpers} takes $\bigO(\mu \log n)$ rounds.
\end{lem}

\begin{proof}
	The round complexity is the sum of the runtime for computing of the ruling set and the local explorations (by local exploration to some depth $d$ we mean that each node forwards all new learned information for $d$ rounds). Each of these steps takes $\bigO(\mu \log n)$ (c.f.\ Lemma \ref{lem:ruling_set}). Next we focus on the correctness.
	
	First we show that $W \in \Theta(pn)$ w.h.p. From this it follows that \smash{$\mu \in \Theta \big(\!\min\big(\sqrt{k}, n/|W|\big)\big)$} (as we require in Definition \ref{def:helpers}). Second we prove that $H_w \geq \mu$ w.h.p., which is condition (1) of Definition \ref{def:helpers}. Third, we show that each node is member of at most $\tilO(1)$ helper sets, corresponding to condition (3). Condition (2) can be concluded immediately from the algorithm, as a node $v \in C_r$ will only join $H_w$ of some $w$ in within the same cluster $C_r$, which has diameter at most $4\mu \lceil\log n\rceil$ hops by construction. 
	
	The random number $|W| \sim Bin(n,p)$ is a sum of $n$ independent Bernoulli variables and the expected value is $\E\big(|W|\big) = pn$. In the following we apply standard Chernoff bounds (c.f.\ Lemma \ref{lem:chernoffbound}\footnote{For the purpose of self-containedness we include the Chernoff- and union bounds we are using in Appendix \ref{apx:generalnotations}}) and use that by asymptotic growth properties, for any constant $c>0$ there is a $n_0 \in \mathbb{N}$, s.t.\ $n^{\delta}> c \log n$ for any const.\ $\delta \in (0,1)$ and all $n \geq n_0$. Then the tail distributions can be bounded as follows
	\begin{align*}
		&\Pr\big(|W| \!<\! \tfrac{pn}{2} \big) \leq e^{-pn/8} = e^{-n^{1-\eps}/8} \,\stackrel{n \geq n_0}{\leq}\, e^{-c \log n} = \tfrac{1}{n^c},\\
		&\Pr\big(|W| \!>\! 2pn \big) \leq e^{-pn/3} = e^{-n^{1-\eps}/3} \,\stackrel{n \geq n_0}{\leq}\, e^{-c \log n} = \tfrac{1}{n^c}.
	\end{align*}
	
	As outlined before, we have that $|C_r| \geq \mu$ due to the minimum distance of $2\mu\!+\!1$ between any two rulers from $\mathcal{R}$ and since nodes join their closest ruler and $G$ is connected. Let $w \in W \cap C_r$. If $q = 1$ in the algorithm, all $v \in C_r$ join $H_w$, hence $|H_w| \geq \mu$. Otherwise we have $q = 2\mu/|C_r|$. Now $|H_w| \sim Bin(|C_r|,q)$ is a sum of independent random variables over $\{0,1\}$ and $\E(|H_w|) = 2\mu \leq 2/p$. Again we can bound the tail distribution of $|H_w|$ (using Lemma \ref{lem:chernoffbound}):	
	\begin{align*}
	\Pr\big(|H_w| \!<\! \E(|H_w|)/2 \big) & \leq \Pr\big(|H_w| \!<\! \tfrac{1}{p} \big) \leq e^{-1/8p}\\ 
	& = e^{-n^{\eps}/8} \,\stackrel{n \geq n_0}{\leq}\, e^{-c \log n} = \tfrac{1}{n^c}
	\end{align*}
	and therefore we have $|H_w| \geq \E(|H_w|)/2 = \mu$ w.h.p. 
	
	It remains to be shown that no node joins too many helper sets. 	
	Let $X_v$ be the random number of helper sets $H_w$ for $w \in W \cap C_r$ that $v$ joins. Then $X_v \sim Bin\big(|W \cap C_r|, q\big)$ is again a sum of independent Bernoulli variables. The expectation is $\E(X_v) = pq \,|C_r| = 2p\mu \leq 2$. We bound the probability that $X_v$ is at most an additive term $3c\log n$ higher than its expectation (with Lemma \ref{lem:chernoffbound}).	
	\begin{align*}
	\Pr\big(X_v \!>\! \E(X_v)\!+\!3c\log n \big) & = \Pr\big(X_v \!>\! \big(1 \!+\! \tfrac{3}{2}c\log n\big)\E(X_v) \big) \\
	& \leq e^{- c \log n} = \tfrac{1}{n^c}.
	\end{align*}
	Finally, the total number of events we showed to occur w.h.p. above is polynomial in $n$. Using Lemma \ref{lem:unionbound} in Appendix \ref{apx:generalnotations} (``union bound'') we can guarantee that \textit{all} of the above events occur w.h.p.
\end{proof}

\end{longversion}

\subsection{Token Routing Protocol}

With the sub-protocol to compute helper sets in place, the main challenge is to pair up helpers $H_s$ of senders $s \in S$ to the according helpers $H'_r$ of receivers $r \in R$. This means that for each token that has to go from $s \in S$ to some $r \in R$ we must pair up two helpers from $H_s$ and $H'_r$ that do the actual transmission. The challenge is to ensure that the pairing has a balanced degree on each side, i.e., the number of tokens each helper has to send or receive is roughly balanced.
The difficulty comes from the fact that we can not rely on global information about helper sets, since broadcasting this would incur too much contention on the global network. Instead, we route messages via intermediate nodes determined through a randomly seeded hash function. 




In order to keep our algorithms reasonably concise, we make a few assumptions. First, we require that each token has a \textit{label} of the form $(s,r,i)$ where $s$ is the sender $r$ is the receiver and $i \in [x]$ enumerates all tokens from $s$ to $r$, where $x \leq k_S$ is the total number of tokens from $s$ to $r$. Second, each receiver initially knows the set of \textit{labels} of tokens it must receive. Since the receivers are required to know how many tokens they must receive from each sender, each receiver can compute the set of labels from each sender.



The last assumption is that each node knows a randomly seeded hash function $h: V \!\times\! V \!\times\! \mathbb{N} \to V$ (from a $k$-wise independent family\lng{, c.f. Definition \ref{def:hashfunctions}}), which maps each token label $(s,r,i)$ to an intermediate node with ID $h(s,r,i)$. We will show that a random seed of size $\tilO(1)$ bits suffices for our purposes, which can be broadcast in $\tilO(1)$ rounds (c.f., Lemma \ref{lem:tokenDissemination}). \lng{Algorithm \ref{alg:token-routing} gives an overview of the involved subroutines.} 

\begin{shortversion}
	In the following we give a consolidated description of the token routing protocol. We outline the correctness and running times along the way to illustrate the proof of Theorem \ref{thm:token-routing}. Comprehensive pseudo-codes and detailed proofs are provided in the full version of this paper \cite{KS20}.
	
	\paragraph*{Routing preparation} Assume that senders $S \subseteq V$ and receivers $R \subseteq V$ were sampled  with probabilities $p_S,p_R$. Then we define \smash{$\mu_S := \lfloor\min\big(\sqrt{k_S}, 1/p_S\big) \rfloor$} and \smash{$\mu_R := \lfloor\min\big(\sqrt{k_R}, 1/p_R\big) \rfloor$} and compute helper sets $H_s, s\in S$ with parameter $\mu_S$ and $H'_r, r\in R$ with parameter $\mu_R$, as outlined in the previous subsection.
	
	Next, we conduct an exploration in the local network for $2(\mu_S\!+\!\mu_R)\lceil \log n \rceil$ rounds (the maximum diameter of any helper set) so that each sender or receiver knows their respective helpers. Based on the knowledge of their helper sets, each sender and receiver assigns their tokens (and token labels) to their helpers in a balanced manner, that is at most \smash{$\lceil \frac{k_S}{\mu_S} \rceil$} tokens or \smash{$\lceil \frac{k_R}{\mu_R} \rceil$} token-labels per helper of a sender or receiver, respectively. Then the token(-labels) are tagged with the according helper and flooded to depth $2(\mu_S\!+\!\mu_R)\lceil \log n \rceil$ in the local network, where each helper simply picks the token(-labels) with its tag on it.
	
	\paragraph*{Routing Scheme} Now tokens are send from helpers of senders to helpers of receivers over the global network. This happens in an indirect manner. First, each sender-helper sends each token $t_{(s,r,i)}$ to the intermediate node with ID $h(r,s,i)$ via the global network. It takes \smash{$\tilO\big(\frac{k_S}{\mu_S})$} rounds until all tokens are at intermediate nodes. Subsequently, each receiver-helper $u$ sends its labels of the form $(s,r,i)$ and their own ID$(u)$ to the intermediate node with ID $h(s,r,i)$, which acts as a request for the token with this label.
	
	The intermediate nodes forward the according tokens to the sender-helpers that requested them. This takes another \smash{$\tilO\big(\frac{k_R}{\mu_R}\big)$} rounds. By using  a hashfunction $h$ from a $k$-wise independent universal family, we can ensure that no node is intermediate of too many tokens, and also that in each round, each intermediate receives at most $\bigO(\log n)$ tokens or requests for tokens (specifics are given in the full paper \cite{KS20}.)
	
	Finally, all tokens that are destined to some receiver $r \in R$ are collectively hold by its helpers $H'_r$. With another $2\mu_R\lceil \log n \rceil$ rounds of flooding (i.e., to the maximum diameter of $H'_r$), all tokens are forwarded from $H'_r$ to $r$.
	
	\paragraph*{Summary} We outline the proof of Theorem \ref{thm:token-routing}, i.e., that token routing can be solved in \smash{$\tilO\big(K/n +\, \sqrt{k_S} +\, \sqrt{k_R}\big)$} rounds.
	The correctness boils down to showing that during the algorithm no messages that are send via the global network are dropped by exceeding the receive bound of $\tilO(\log n)$ messages. This is shown in the full version of this article \cite{KS20}.
	
	It remains to sum up the running times. The first task is the computation of helper sets which takes $\tilO(\mu_s \!+ \!\mu_r)$ rounds as outlined in the previous subsection. Assigning tokens and labels to helpers via flooding takes $\tilO(\mu_s \!+ \!\mu_r)$ rounds as well. The final routing process among helpers via the global network and subsequent collection of tokens via the local network takes \smash{$\tilO\big(\frac{k_S}{\mu_S} \!+\! \frac{k_R}{\mu_R} \!+\! \mu_R\big)$} rounds. The total runtime is \smash{$\tilO\big(\frac{k_S}{\mu_S} \!+\! \frac{k_R}{\mu_R} \!+\! \mu_S \!+\! \mu_R \big)$}. 
	
	Recall that we chose $\mu_W = \lfloor\min\big(\sqrt{k_W}, 1/p_W\big) \rfloor$ for $W = {S,R}$. We have that $|S| = \Theta(n\cdot p_S)$ and $|R| = \Theta(n \cdot p_R)$ w.h.p. Thus $\mu_S \in \Theta\big( \min\big(\sqrt{k_S}, n/|S|\big) \big)$ and $\mu_R \in \Theta\big( \min\big(\sqrt{k_R}, n/|R|\big) \big)$. Therefore the runtime can be bounded by \smash{$\tilO\big(\frac{|S|\cdot k_S}{n} + \frac{|R|\cdot k_R}{n} +\, \sqrt{k_S} +\, \sqrt{k_R} \big)$}.
\end{shortversion}

\begin{longversion}

\begin{algorithm}[H]
	\caption{\texttt{Token-Routing}$(p_S,p_R)$}
	\label{alg:token-routing}
	\begin{algorithmic}
		\State {$\mu_S \gets \lfloor\min\big(\sqrt{k_S}, 1/p_S\big) \rfloor$}; \smash{$\mu_R \gets \lfloor\min\big(\sqrt{k_R}, 1/p_R\big) \rfloor$} 
		\State \texttt{Compute-Helpers}($\mu_S$); \texttt{Compute-Helpers}($\mu_R$) 
		\State \texttt{Routing-Preparation}($\mu_S, \mu_R$) \Comment{\textit{assign tasks to helpers}}
		\State \texttt{Routing-Scheme} \Comment{\textit{forward tokens via intermediate nodes}}
	\end{algorithmic}
\end{algorithm}

The following subroutine \texttt{Routing-Preparation} summarizes the necessary steps to prepare the final protocol \texttt{Routing-Scheme}. Each sender and receiver determines its respective helper set. Then every sender balances its tokens it must send and every receiver balances the labels of tokens it must receive among their helpers in time $\tilO(\mu_S \!+\! \mu_R)$.

\begin{algorithm}[H]
	\caption{\texttt{Routing-Preparation}$(\mu_S,\mu_R)$}
	\label{alg:routing-preparation}
	\begin{algorithmic}
		\For {$2(\mu_S\!+\!\mu_R)\lceil \log n \rceil$ rounds} \Comment{\textit{detect helpers via flooding}}
		\State $v$ forwards information on $H_w,H'_w$ via incident edges
		\EndFor
		\If {$v \in S$ (or $v \in R$)}
		\State compute \textit{balanced} assignment of tasks to helpers $H_v$ ($H'_v$)  
		\EndIf
		\For {$2(\mu_S\!+\!\mu_R)\lceil \log n \rceil$ rounds} \Comment{\textit{flood token(-labels) to helpers}}
		\State $v$ picks token(-labels) assigned to it
		\State $v$ forwards all other token(-labels) to local neighbors
		\EndFor
	\end{algorithmic}
\end{algorithm}

\begin{fact}
	\label{fct:routing-preparation}
	Assume that helper sets $H_s, s \in S$, $H'_r, r \in R$ have already been established for parameters $\mu_S,\mu_R$. Then by executing Algorithm \ref{alg:routing-preparation}, all tokens of any sender $s \in S$ and all token-labels of any receiver $r \in R$ are distributed among their respective helpers via the local network, such that each helper $w \in H_s$ has at most \smash{$\lceil \frac{k_S}{\mu_S} \rceil$} tokens from $s$ and each $w \in H'_r$ has at most \smash{$\lceil \frac{k_R}{\mu_R} \rceil$} token-labels from $r$. The assignment of token(-labels) to helpers can be computed by each sender and receiver after learning its helper set (in the first loop). 
	Algorithm \ref{alg:routing-preparation} takes $\bigO\big((\mu_S\!+\!\mu_R) \log n \big)$ rounds, due to Lemma \ref{lem:helpers}.\footnote{We will only briefly summarize claims that we believe to be evident or easily verifiable by the reader, and refer to those claims as ``Fact''.}
\end{fact}

The final sub-procedure \texttt{Routing-Scheme} uses the hash function $h: V \times V \times \mathbb{N} \to V$ to route each token $t_{(s,r,i)}$ via the node $h(s,r,i)$. Essentially, each helper $w \in H_s$ of a sender $s \in S$ sends its tokens to the respective intermediate nodes given by the hash-value of its label, $O(\log n)$ tokens at a time until all tokens are at intermediate nodes. Subsequently, we do the same in reverse on the receiver side. Each helper $w \in H'_r$ of a receiver $r \in R$, sends requests for its assigned token labels ($\bigO(\log n)$ per round)  to the respective intermediate nodes. The intermediate nodes forward all requested tokens in the subsequent round.

\begin{algorithm}[h]
	\caption{\texttt{Routing-Scheme} \Comment{\textit{$\!h\!:\! V \!\times\! V \!\times\! \mathbb{N} \!\to\! V$ globally known}}}
	\label{alg:routing-scheme}
	\begin{algorithmic}
		\State $T \gets$ tokens $v$ must send	\Comment{\textit{tasks of $v$ as sender-helper}}
		\While{$T \neq \emptyset$}
			\State \texttt{out} $\gets $ pick $\bigO(\log n)$ tokens from $T$; 		
			\State send each $t_{(s,r,i)} \in$ \texttt{out} to node $h(s,r,i)$ via global network
			\State $T \gets T \setminus$\texttt{out}
		\EndWhile	
		\State $I \gets$ tokens $v$ received as intermediate node 
		\State $T' \!\gets$ \textit{labels} of tokens $v$ must receive \Comment{\textit{tasks as receiver-helper}}
		\While{$T'\neq \emptyset$ or $I \neq \emptyset$}
			\State \texttt{req} $\gets $ pick $\bigO(\log n)$ token \textit{labels} from $T'$;
			\State for each ${(s,r,i)} \in \texttt{req}$ send $(s,r,i)$ and $ID(v)$ to node $h(s,r,i)$ 
			\State $T' \gets T' \setminus$\texttt{req}
			\State \texttt{rcv} $\gets $ all requests received
			\State for each $(s,r,i) \in$ \texttt{rcv} received from some $u$, send $t_{(s,r,i)}$ to $u$
			\State $I \gets I \setminus \{t_{(s,r,i)} \mid {(s,r,i)} \in$ \texttt{rcv}$\}$
		\EndWhile	
		\For {$2\mu_R\lceil \log n \rceil$ rounds} \Comment{\textit{collect tokens from helpers}}	
			\State $v$ forwards tokens to its neighbors via local edges
		\EndFor		
	\end{algorithmic}
\end{algorithm}

\begin{lem} 
	\label{lem:routing-scheme}
	During the execution of Algorithm \ref{alg:routing-scheme} no node receives more than $O(\log n)$ messages per round. After Algorithm \ref{alg:routing-scheme} has terminated, every receiver knows all tokens that it was target of. A random hash function $h$ that guarantees the above, requires a random seed of $\bigO(\log^2 n)$ bits. Algorithm \ref{alg:routing-scheme} takes \smash{$\tilO\big(\frac{k_S}{\mu_S} + \frac{k_R}{\mu_R} + \mu_R\big)$} rounds. 
\end{lem}

\begin{proof}
	We start with the runtime. By Fact \ref{fct:routing-preparation}, each node has at most \smash{$\lceil \frac{k_S}{\mu_S} \rceil$} tokens to send and at most \smash{$\lceil \frac{k_R}{\mu_R} \rceil$} tokens to receive. Every node receives $\bigO(\log n)$ messages per round (as we argue below), thus every node can answer all requests for tokens they receive in their role as an intermediate node directly in the subsequent round. Therefore the process of funneling tokens from sender-helpers to receiver-helpers via intermediates takes \smash{$\tilO\big(\frac{k_S}{\mu_S} \!+\! \frac{k_R}{\mu_R}\big)$} rounds. It takes another $\tilO(\mu_R)$ rounds until receivers have collected all tokens from their helpers in the final loop.
	
	The algorithm is correct if each receiver eventually obtains all the tokens it must receive. That is no message must be dropped during the runtime by exceeding the receive bound of $O(\log n)$ messages per node and round. In Lemma \ref{lem:receiveBound} we show that this is guaranteed by selecting targets for tokens (or token requests) is a uniform and $k$-wise independent random process, for some $k \in \Theta(\log n)$.
	
	Therefore, we source our hash function $h: V \times V \times \mathbb{N} \to V$ from a $k$-wise independent family $\mathcal H$ as defined in Definition \ref{def:hashfunctions}. That such a family exists and that selecting a hash function from it requires $\bigO(\log^2 n)$ random bits is implied by Lemma \ref{lem:hashfunctions}. This guarantees that selecting targets by applying the hash function $h(s,r,i)$ on token labels $(s,r,i)$ is $k$-wise independent if the labels are distinct (which they are).
\end{proof}

We prove Theorem \ref{thm:token-routing} by showing that Algorithm \ref{alg:token-routing} solves token routing in $\tilO\big(K/n +\, \sqrt{k_S} +\, \sqrt{k_R}\big)$ rounds.

\begin{proof}[Proof of Theorem \ref{thm:token-routing}]
	The correctness of Algorithm \ref{alg:token-routing} was shown in Lemma \ref{lem:routing-scheme} and what remains is to sum up the runtime of the sub-procedures. 
	The first task is the computation of helper sets in Algorithm \ref{alg:compute-helpers} which takes $\tilO(\mu_s \!+ \!\mu_r)$ rounds due to Lemma \ref{lem:helpers}. Assigning tokens and labels to helpers in Algorithm \ref{alg:routing-preparation} also takes $\tilO(\mu_s \!+ \!\mu_r)$ rounds, as outlined in Fact \ref{fct:routing-preparation}. The final routing process among helpers via the global network and subsequent collection of tokens via the local network can be done in \smash{$\tilO\big(\frac{k_S}{\mu_S} \!+\! \frac{k_R}{\mu_R} \!+\! \mu_R\big)$} rounds, as shown by Lemma \ref{lem:routing-scheme}. 
	
	The total runtime is \smash{$\tilO\big(\frac{k_S}{\mu_S} \!+\! \frac{k_R}{\mu_R} \!+\! \mu_S \!+\! \mu_R \big)$}. In Lemma \ref{lem:helpers} we showed that for the parameters $\mu_S,\mu_R$ computed in Algorithm \ref{alg:token-routing} we have $\mu_S \in \Theta \big(\!\min\big(\sqrt{k_S}, n/|S|\big)\big)$ and $\mu_R \in \Theta \big(\!\min\big(\sqrt{k_R}, n/|R|\big)\big)$. Therefore, the total runtime is \smash{$\tilO\big(\frac{|S|\cdot k_S}{n} + \frac{|R|\cdot k_R}{n} +\, \sqrt{k_S} +\, \sqrt{k_R} \big)$}.
\end{proof}	

\end{longversion}

	\section{All Pairs Shortest Paths}
\label{sec:apsp}

Our first application of the token routing protocol is the APSP problem. Specifically, we show that APSP can be solved \textit{exactly} in the \HYBRID model on weighted graphs in \smash{$\tilO\big(\sqrt{n}\big)$} rounds. This is an improvement over the $\tilO({n^{2/3}})$ complexity algorithm of \cite{Augustine2020}. In fact, \cite{Augustine2020} gives a \smash{$\tilOm\big(\sqrt{n}\big)$} lower bound for APSP algorithms in the \HYBRID model, even for $\alpha$-approximations for some $\alpha \in \tilT\big(\sqrt{n}\big)$. Consequently, our algorithm settles the APSP problem in the \HYBRID model in terms of runtime (up to $\polylog n$ factors) and approximation factors (i.e.\ admitting approximation factors up to $\tilO\big(\sqrt{n}\big)$ does not allow faster algorithms).

The algorithm presented in this paper reuses most of the routines given by \cite{Augustine2020} with a major rework of just one subroutine that constitutes a runtime bottleneck in the original algorithm. This section shows that the bottleneck can be resolved using the token routing protocol of Section \ref{sec:token-routing}. In order to keep this section reasonably concise, we refer to \cite{Augustine2020} for full proofs and detailed pseudo-codes of the base algorithm. In order to preserve some measure of self-containedness, we summarize the exact APSP algorithm of \cite{Augustine2020} in the following few paragraphs (supplemented by some key Lemmas given in Appendix \ref{apx:skeleton-graphs}), therefore knowledge of \cite{Augustine2020} is not strictly required to understand this section.

As in \cite{Augustine2020}, we leverage the concept of a skeleton graph $\mathcal S = (V_\calS,E_\calS)$ (first introduced by \cite{Ullman1991} in the context of the PRAM model). The idea is to sample a subset $V_\calS \subseteq V$ randomly and compute shortest paths among sampled nodes $V_\calS$ that are in relatively close hop distance of each other. These shortest paths form the edge set $E_\calS$ of $\calS$. 
A key insight is that if we sample each node into $V_\calS$ with probability $\frac{1}{x}$ for some $x \in [n]$, then there will be a sampled node on a shortest path $P_{u,v}$ between any pair $u,v \in V$ roughly every $\tilO(x)$ hops (c.f., Lemma \ref{lem:long-paths-skeleton}). 

If we set the hop-length of paths in $E_\calS$ to some value $h \in \tilT(x)$ (the correct value of $h$ is given in Lemma \ref{lem:long-paths-skeleton}), then this implies that for any pair $u,v \in V$ with sufficient hop distance $hop(u,v) \geq h$ a shortest $u$-$v$-path $P_{u,v}$ must run over the skeleton (c.f., Lemma \ref{lem:long-shortest-paths-skeleton}). Then the length $w(P_{u,v})$ of $P_{u,v}$ can be computed by adding the length of the sub-paths $w(P_{u,v}) = w(P_{u,s_1}) + w(P_{s_1,s_2}) + w(P_{s_2,v})$, where $s_1$ and $s_2$ are the first and the last skeleton node on $P_{u,v}$, respectively.

The computation of the distances among skeleton nodes $s_1,s_2 \in V_\calS$ relies on the \textit{global} knowledge of the skeleton $\calS$. To achieve that, each skeleton node determines its neighbors in $\calS$ using the local network. This can be done in $\tilO(x)$ rounds as edges $E_\calS$ correspond to paths in $G$ with at most $h \in \tilO(x)$ hops. Since $\calS$ will have much fewer nodes than $G$ w.h.p.\ (namely $\tilO\big(\frac{n}{x}\big)$) the size of $E_\calS$ is manageable and we can make $E_\calS$ globally known using the \textit{token dissemination protocol} of \cite{Augustine2020} (replicated in Lemma \ref{lem:tokenDissemination}). Specifically, the $\tilO\big((n/x)^2\big)$ edges of $E_\calS$ can be made public knowledge in $\tilO\big(\frac{n}{x}\big)$ rounds. 

Subsequently, the length of shortest paths $P_{s_1,s_2}$ for any $s_1,s_2 \in V_\calS$ can be computed locally by all nodes in the network.
Moreover, any node $u \in V$ can compute its distance to any $s_2 \in V_\calS$ by computing the minimum of $w(P_{u,s_1}) + w(P_{s_1,s_2})$, where $s_1 \in V_\calS$ are the skeleton nodes in $u$'s vicinity (that is within $h \in \tilO(x)$ hops). The latter can be determined fast in the local network.
Now that any node is able to compute its own distance to all skeleton nodes, the part that remains is to compute distances to \textit{any} node in $V$ that has hop distance at least $h$. This is arguably the hardest part. In order that $u \in V$ can learn its distance to some $v \in V$, it must learn over which skeleton node $s_2 \in V_\calS$ in $v$'s vicinity a shortest path $P_{u,v}$ runs, as well as the length of the corresponding sub-path $P_{s_2,v}$. 

In \cite{Augustine2020} this is done by another run of \textit{token dissemination} where each token corresponds to a distance label $d_h(s,v)$ for $(s,v) \in V_\calS \times V$. Having all these distance labels available at every node in the network makes computing APSP simple, as any node $u \in V$ can now compute $w(P_{u,s_1}) \!+ w(P_{s_1,s_2}) \!+ w(P_{s_2,v})$ for any $v \in V$ and any $s_1, s_2 \in V_\calS$ within $h$ hops of $u$ and $v$ respectively (the minimum over $s_1, s_2 \in V_\calS$ yields the true distance $d(u,v)$). However, this approach requires that \smash{$|V|\!\cdot\!|V_\calS| \in \tilT(\frac{n^2}{x}) \gg n$} distance labels are broadcast in the worst case, taking $\tilT\big(n/\,\sqrt{x}\big)$ rounds. This distorts the trade-off between local and global computation and pushes the runtime up to $\tilT(n^{2/3})$ (the total runtime including the last step amounts to $\tilO\big(x + n/\sqrt{x}\big)$ which is optimized for $x = n^{2/3}$).

Note that the above procedure actually broadcasts more information than necessary, which one can see from the fact that afterwards, \textit{every} node can compute the distance between \textit{every} pair $u,v \in V$ (as long as $u$ and $v$ have hop distance at least $h$). However, in the distributed version of the APSP problem, a given node is only required to learn its \textit{own} distance to the other nodes. This is where the {token routing} protocol from Section \ref{sec:token-routing} can help.


\begin{proof}[Proof of Theorem \ref{thm:APSP}]
	Suppose we do the same steps as described above, except for the last, where distance labels $d_h(s,v)$ for $(s,v) \in V_\calS \times V$ are broadcast in the \textit{whole} network with the token dissemination protocol. Up to this point, every node $v \in V$ knows its distance $d(v,s)$ to every skeleton node $s \in V_\calS$ and the first $s'\in V_\calS$ on a shortest path $P_{v,s}$. We call $s'$ a \textit{connector} for $v$ to $s$ (note that $hop(v,s') \leq h$ and $s'=s$ is possible). For the proof of correctness and runtime up to this point we refer to our explanations above and to \cite{Augustine2020} for a detailed analysis. 	
	
	Now each node $v\in V$ creates a token $t_{(v,s)}$ with sender $v$ and receiver $s$ containing $\langle d_h(v,s'), ID(v), ID(s') \rangle$, i.e.\ the distance from $v$ to the connector $s'$ for the path from $v$ to $s$. This creates an instance of the token routing problem, where senders correspond to $V$ and the receivers correspond to $V_\calS$. Each sender has exactly one token for each receiver, that is $k_S = |V_\calS| \in \tilO\big(\frac{n}{x}\big)$. Conversely, each receiver expects $k_R = n$ tokens, precisely one from each node.
	
	Solving this instance of token routing gives each skeleton node $s \in V_\calS$ enough information to compute its distance to every node $v \in V$ by locally computing $d(s,s') + d_h(s',v)$ where $s' \in V_\calS$ is the connector given in the token $t_{(v,s)}$, which also contains $d_h(v,s')$.	
	In another $\tilO(x)$ rounds of local communication, each skeleton nodes can distribute distance labels $\langle d(s,v), ID(s), ID(v)\rangle$ to all nodes in their $\tilO(x)$ neighborhood via the unrestricted local network. Equipped with the information from these labels, each node $u \in V$ can compute its distance to each $v \in V$ by taking the minimum of $d_h(u,s) + d(s,v)$ for each $s \in V_\calS$ in the $h$-hop neighborhood of $u$.
	
	The steps of the original algorithm (excluding the last step) take $\tilO(x)$ rounds for the local communication up to hop-distance $h$ and $\tilO\big(\frac{n}{x}\big)$ for making the skeleton public knowledge. By Theorem \ref{thm:token-routing}, solving the above instance of token-routing takes 
	\[
		\tilO\big(\tfrac{n \cdot (n/x) + (n/x) \cdot n}{n} +\, \sqrt{|V|} +\, \sqrt{|V_\calS|}\big) = \tilO\big(\tfrac{n}{x}+ \sqrt{n}\big)
	\]
    rounds. This results to a total runtime of \smash{$\tilO\big(x + \frac{n}{x}+ \sqrt{n}\big) = \tilO\big(\sqrt{n}\big)$} by choosing $x = \sqrt{n}$.
\end{proof}

	\section{$\mathbf k$-Source Shortest Paths}
\label{sec:cc-sp-sim}

Next, we apply the token routing algorithm to simulate the \textit{Congested Clique} model (or \CC model) on a subset of nodes in a \HYBRID network. This allows us to simulate powerful algorithms from \cite{Censor-Hillel2019a,Censor-Hillel2019b} for the \CC model and transfer their results for shortest paths (this section) and diameter problems (subsequent section) into the domain of the \HYBRID model. In order to be efficient, we provide a general framework that allows to translate \CC algorithms that are parametrized by runtime complexity, number of sources and multiplicative and additive approximation parameters into \HYBRID algorithms.
A subset of results we obtain that way are tight in terms of runtime (bar $\polylog n$ factors), which we demonstrate in \shrt{the final Section \ref{sec:lower-bounds}}\lng{Section \ref{sec:k-src-lb}}.

We start by defining the \CC model, which is a member of the family of synchronous message passing models (like the \HYBRID model). Nodes possess unlimited computing power, unique IDs and exchange messages in synchronous rounds. In a \CC each node is allowed to send one (different) $\bigO(\log n)$ bit message to \textit{every} node each round.\lng{\footnote{Using the constant time routing scheme from the \CC model by Lenzen \cite{Lenzen13} one can even assume that in the \CC model, each node can send $n$ messages to arbitrary targets, i.e, allowing multiple messages to the same node.}}
The following corollary of Theorem \ref{thm:token-routing} demonstrates the connection between \CC and \HYBRID networks that we aim to exploit in this section.

\begin{cor}
	\label{cor:cc-sp-sim}
	Let $S \subseteq V$ be a subset of nodes from a \HYBRID network, obtained by sampling each node with probability $\frac{1}{n^{1-{x}}}$ for ${x} \in (0,1]$. One round of the \CC model can be simulated on $S$ in \smash{$\tilO\big(n^{2{x}-1} \!+ n^{{x}/2}\big)$} rounds w.h.p.
\end{cor}

\begin{proof}
	Each node $v \in S$ is a sender and a receiver of exactly $|S|$ messages, one to and one from each node. We make $S$ public knowledge with a run of token dissemination in \smash{$\tilO\big(\sqrt{|S|}\big)$} (c.f. Lemma \ref{lem:tokenDissemination}), so each node in $S$ knows from whom it receives a message. Then one round of \CC on $S$ corresponds to an instance of the token routing problem. Hence the result follows immediately from Theorem \ref{thm:token-routing}. In the simulation of the \CC model, each $v \in S$ must send/receive at most one message to/from every other node in $S$, i.e. at most $|S|$ messages in total. This takes \smash{$\tilO\big({|S|^2}/{n} \!+ \sqrt{|S|}\big)$}. Using Chernoff bounds\lng{ (Lemma \ref{lem:chernoffbound})} we have $|S| \in \Theta(n^{{x}})$ w.h.p.
\end{proof}

Based on this simulation result we can provide the framework to transform shortest path algorithms in the \CC model into algorithms for the \HYBRID model. The following theorem gives the details on how runtime, number of sources and approximation parameters of \CC algorithms can be translated accordingly. \shrt{Further into this section, we describe the simulation algorithm and outline the proof of Theorem \ref{thm:cc-sp-sim}. Detailed proofs and pseudo-codes are provided in the full paper \cite{KS20}.}\lng{Further into this section, we give the detailed simulation algorithm and then prove individual claims in separate lemmas, which, when combined, culminate in the proof of Theorem \ref{thm:cc-sp-sim}.} At the end of this section, we apply the theorem to a series of known \CC algorithms which generates according results for the \HYBRID model. We remark that future advances in \CC algorithms for shortest paths may give further improvements for the \HYBRID model.\phil{Mention that the new \CC paper that solves shortest paths in loglog(n) does not help?}

\begin{thm}
	\label{thm:cc-sp-sim}
	Let $\calA$ be an $(\alpha,\beta)$-approximation \CC algorithm that computes \emph{weighted} shortest paths for $n^{\gamma}, \gamma \in [0,1]$ sources in time \smash{$T_\calA \in \tilO(\eta n^\delta)$} for constants $\delta \geq 0, \eta \geq 1$. Let \smash{${x} := \frac{2}{3 + 2 \delta}$}.
	
	Then there is a \HYBRID algorithm $\calB$ that solves the $n^{{x}\gamma}$ source shortest paths problem with runtime $T_\calB \in \tilO\big(\eta n^{1-{x}}\big)$ and the following additional properties. On \emph{weighted} graphs, $\calB$ computes a $(2\alpha \!+\! 1 \!+\! {\beta}/T_\calB)$-approximation.
	On \emph{unweighted} graphs, $\calB$ computes an \smash{$\big(\alpha \!+\! \tfrac{2}{\eta} \!+\! \tfrac{\beta}{T_\calB}\big)$}-approximation.
	For $\gamma = 1$ ($\calA$ solves APSP), we can do the above for an arbitrary number of $k \in [n]$ sources but this adds $\tilO\big(\sqrt{k}\big)$ rounds to the runtime $T_\calB$. For $\gamma = 0$ (SSSP) the approximation factor improves to $(\alpha\!+\! {\beta}/T_\calB)$ for weighted and unweighted graphs. 
\end{thm}

\subsection{Simulation Algorithm}

The intuition behind the algorithm is as follows. Initially, we compute a skeleton graph $\calS =(V_\calS, E_\calS)$ in the \HYBRID network by first sampling nodes to $V_\calS$ with probability $\frac{1}{n^{1-{x}}}$. After the sampling, we use the local network to determine the skeleton edges $E_\calS$, which represent paths in $G$ with hop distance at most $h \in \tilO(n^{1-{x}})$. The parameter ${x}$ determines the size of the skeleton (i.e., $|V_\calS| \in  \tilO(n^{x})$ w.h.p.) and optimizes the trade-off between the time required for the simulation of $\calA$ on $\calS$ and the time spent for communication on the local network (e.g., to compute $E_\calS$).
This establishes a fundamental property of skeleton graphs, namely that on a shortest path between two endpoints more than $h$ hops apart there must be a skeleton node at least every $h$ hops. Moreover, a shortest path must run over the skeleton via two skeleton ``entry-nodes'', each within $h$ hops to one of the endpoints (we outlined this earlier in Section \ref{sec:apsp} and formalize it in Lemmas \ref{lem:long-paths-skeleton},\ref{lem:long-shortest-paths-skeleton}).

Note that at this point each skeleton node knows only its incident edges in $\calS$. In contrast to the APSP algorithm in Section \ref{sec:apsp}, we do \textit{not} make $\calS$ public knowledge, since this would take too long with our available tools. Instead, the local knowledge of skeleton edges by incident skeleton nodes serves as input for the corresponding graph problem on $\calS$, which we solve by simulating an appropriate \CC algorithm $\calA$ in $\calS$ as outlined in Corollary \ref{cor:cc-sp-sim}. Afterwards, the distances to the set of sources (or estimations thereof) are known to all skeleton nodes. But since every node in $V$ has a skeleton node within $h$ hops on a shortest path to a source (as formalized by Lemma \ref{lem:long-paths-skeleton}), they can simply gather the required information to compute their distance to all sources via the local network.

There is one catch however, as sources in a graph problem given on the {whole} network $G$ will probably not coincide with the randomly sampled skeleton nodes $V_\calS$. 
Generally speaking, we can not simply add all sources from $G$ to the skeleton since they might be too regionally concentrated in $G$, thereby obstructing the creation of helper sets, where each node must only help $\tilO(1)$ senders or receivers, which would also inhibit the (fast) simulation of the \CC. Note however, that if we are dealing with only a single source in $G$, adding it to $V_\calS$ is actually not a problem as each helper will be assigned at most one additional skeleton in the simulation routine (c.f., Corollary \ref{cor:cc-sp-sim}). 

In general, we circumvent the mentioned problem by letting each source $v \in V$ that has not been sampled to $V_\calS$ tag its closest skeleton $r_v \in \calS$ as its \textit{representative}, i.e.\ as substitute source in $V_\calS$. Subsequently, we make representatives and distances to their original source publicly known (a comparatively small amount of information). After computing distances to the representatives as described above, each node can estimate its distance to the an original source by adding its distance to the according representative to the (publicly known) distance from representative to the original source.
The drawback is that the estimations will increase the additive error by up to $2h$ (two times the maximum hop-length of a skeleton edge) in the unweighted case and the multiplicative error by $(1\!+\!\alpha)$ in the weighted case. We show that additive errors of $\Theta(h)$ are in fact not that troublesome. On one hand, we can use the local network to compute distances up to some arbitrarily large multiple of $h$ \textit{exactly}. On the other hand, the multiplicative impact of additive errors on paths longer than that is relatively small.

\begin{shortversion}
The detailed proof of Theorem \ref{thm:cc-sp-sim} is given in the full article \cite{KS20}. In a nutshell, the proof of the runtime consists of summing up the runtimes of the simulation of $\calA$ and the local exploration to compute small distances exactly. Note that for the chosen $x = \frac{2}{3 + 2 \delta}$ one exponent dominates the runtime \smash{$\tilO\big(n^{2{x}-1} \!+ n^{{x}/2}\big)$} of the simulation given in Corollary \ref{cor:cc-sp-sim}, therefore the runtime of simulating one round simplifies to $\tilO\big(n^{{x}/2}\big)$. Multiplied with the runtime of $\calA$ we have $\tilO\big(\eta  n^{x\delta + x/2}\big)$ rounds to simulate $\calA$. The local communication sums up to $\tilO(\eta  n^{1-x})$. The trade-off between simulation and local communication is optimized for $x := \tfrac{2}{3+2 \delta}$ making the running times of both steps equal to \smash{$\tilO(\eta  n^{1-x})$}.

Although lengthy, proving the claimed approximation ratios is not that technically involved. It boils down to proving the error inflicted by computing shortest paths to representatives instead of directly to the sources. We rely extensively on the properties of skeleton graphs (given in Lemmas \ref{lem:long-paths-skeleton} and \ref{lem:long-shortest-paths-skeleton}) and the triangle inequality. Subsequently, we use the exact knowledge of distances of paths with few hops to turn additive errors into small multiplicative errors. We refer the interested reader to the full paper \cite{KS20}.
\end{shortversion}

\begin{longversion}

\begin{algorithm}[H]
	\caption{\texttt{SP-Simulation($\calA, \gamma, \delta, \eta$)} \Comment{\textit{param.\ of Theorem \ref{thm:cc-sp-sim}}}}
	\label{alg:cc-sp-sim}
	\begin{algorithmic}
		\State ${x} \gets \frac{2}{3 + 2 \delta}$ \Comment{\textit{${x}$ optimizes overall runtime}}
		\State \texttt{Compute-skeleton($\gamma, {x}$)} \Comment{\textit{construct $\calS \!=\!(V_\calS, E_\calS)$}}
		\State \texttt{Compute-Representatives} \Comment{\textit{pick representatives in $V_\calS$}}
		\State \texttt{Clique-Simulation($\calA,x$)} \Comment{\textit{run $\calA$ on $\calS$}}
		\For{$\eta  h$ rounds} \Comment{\textit{flood dist.\ to sources locally}}
			\State $v$ forwards distances to sources via local edges
		\EndFor
		\State $v$ computes own distance to all sources with Equation \eqref{eq:cc-sp-sim}
	\end{algorithmic}
\end{algorithm}

\begin{fact}
	\label{fct:cc-sp-sim}
From the simulation of $\calA$ and the subsequent local distribution of the results in Algorithm \ref{alg:cc-sp-sim}, each node $v \in V$ knows an $(\alpha,\beta)$-estimation $\tild(u,r_s)$ between each skeleton $u \in V_\calS$ and each representative $r_s$ of some source $s$. Additionally each  node knows the distance $d_h(r_s,s)$ between each source $s$ and its representative $r_s$ from the respective token dissemination in the sub-procedure given by Algorithm \ref{alg:compute-skeleton}. If two nodes $u,v \in V$ are within $h$ hops of each other they know their $\eta  h$-hop limited distance $d_{\eta  h}(u,v)$ from the local exploration. With the available information every node $v \in V$ can compute an estimation $\tild(v,s)$ to each source $s \in V$ as follows
\begin{equation}
	\label{eq:cc-sp-sim}
	\tild(v,s) := \min \big(d_{\eta  h}(v,s), \min_{u \in V_\calS} d_h(v,u) + \tild(u,r_s) + d_h(r_s,s) \big).
\end{equation}
\end{fact}

We outsource several tasks from the main algorithm to sub-procedures with the intent of reusing some of them in the next section for the diameter problem. The task of computing the skeleton graph is moved to Algorithm \ref{alg:compute-skeleton}. The task of establishing the relations among sources and representatives and making these public knowledge is done by Algorithm \ref{alg:compute-reps}. The simulation routine as outlined in Corollary \ref{cor:cc-sp-sim} is given by Algorithm \ref{alg:cc-sim}.

\begin{algorithm}[H]
	\caption{\texttt{Compute-Skeleton($\gamma,{x}$)} \Comment{\textit{$\!h \!\in\! \tilO(n^{1-{x}})$ see Lem.\ \ref{lem:long-paths-skeleton}}}}
	\label{alg:compute-skeleton}
	\begin{algorithmic}		
		\State $v$ joins $V_\calS$ with probability $\frac{1}{n^{1-{x}}}$ \Comment{\textit{$|V_\calS| \!\in\! \tilT(n^{x})$ w.h.p.}}
		\If {$v$ is a source \textbf{and} $\gamma = 0$ } $v$ joins $V_\calS$
		\EndIf
		\For {$h \in \tilO(n^{1-{x}})$ rounds} \Comment{\textit{locally determine $E_\calS$}}
		\State $v$ forwards all graph information< via local network
		\EndFor
	\end{algorithmic}
\end{algorithm}

\begin{fact}
	\label{fct:compute-skeleton}
	The size of $V_\calS$ is \smash{$\tilT\big(n^{{x}}\big)$} w.h.p. Let $h := \xi n^{1-{x}} \ln n$, where $\xi$ is the parameter given in Lemma \ref{lem:long-paths-skeleton}. Algorithm~\ref{alg:compute-skeleton} establishes a weighted graph $\calS = (V_\calS,E_\calS)$ among the sampled nodes $V_\calS$ in $\tilO\big(n^{1-{x}}\big)$ rounds, whereas we define $E_\calS \coloneqq \{ \{u,v\} \!\mid\! u,v\!\in\!V_\calS, \text{hop}(u,v) \!\leq\! h\}$. The weight of $\{u,v\} \in E_\calS$ is defined as $d_h(u,v)$.
	After the subroutine, all skeleton nodes know their neighbors in $\calS$ and the distances of the incident edges in $E_S$. Further properties of $\calS$ are outlined in Lemma \ref{lem:long-shortest-paths-skeleton}.
\end{fact}

\begin{algorithm}[H]
	\caption{\texttt{Compute-Representatives}}
	\label{alg:compute-reps}
	\begin{algorithmic}		
		\If {$v$ is a source}  
		\State $v$ tags closest node in $V_\calS$ as rep.\ $r_v$ \Comment{\textit{set $v \!=\! r_v$ if $v \!\in\! V_\calS$}}
		\State $v$ creates token $\langle d_h(v,r_v), ID(v),ID(r_v) \rangle$ 
		\State \Comment{\textit{$d_h(v,r_v) = 0$ and $ID(v)=ID(r_v)$ indicates $v \in V_\calS$}}
		\EndIf
		\State $v$ participates in token dissemination protocol
		\If {$v$ received token $\langle \cdot, \cdot ,ID(v) \rangle$} 
		\State $v$ becomes source in $\calS$ \Comment{\textit{$v$ was tagged by some source in $G$}}
		\EndIf
	\end{algorithmic}
\end{algorithm}

\begin{fact}
	\label{fct:compute-reps}
	After Algorithm \ref{alg:compute-reps} terminates, every source $s$ of $G$ has a representative $r_s \in V_\calS$ ($s=r_s$ if $s \in V_\calS$) and $r_s$ and the distances $d_h(s,r_s)$ are known by the whole network for all sources $s$ of $G$.
\end{fact}

\begin{algorithm}[H]
	\caption{\texttt{Clique-Simulation($\calA,x$)} \Comment{\CC alg.\ $\calA$}}
	\label{alg:cc-sim}	
	\begin{algorithmic}
		\For{$i \gets 1$ \textbf{to} $T_\calA$} \Comment{\textit{simulation of $\calA$}}
		\If {$v \in V_\calS$} \Comment{\textit{nodes $V_\calS$ take part in simulation}}
		\State $v$ performs computations based on $\calA$ round $i$
		\State $v$ creates tokens \smash{$t^i_{(v,u)}$} to nodes $u\in V_\calS$ for round $i$		
		\EndIf
		\State \texttt{Token-Routing($\frac{1}{n^{1-{x}}},\frac{1}{n^{1-{x}}}$)} \Comment{\textit{$V \!\setminus\! V_\calS$ take part as helpers}}	
		\EndFor
	\end{algorithmic}
\end{algorithm}

The full proof of Theorem \ref{thm:cc-sp-sim} is given in the following subsection. In a nutshell, the proof of the runtime of Algorithm \ref{alg:cc-sp-sim} consists of summing up the runtimes of the simulation of $\calA$ (which is $\tilO\big(\eta  n^{x\delta + x/2}\big)$) and the local exploration to compute small distances exactly (which is $\tilO(\eta  n^{1-x})$). The trade-off is optimized for the parameter $x := \tfrac{2}{3+2 \delta}$ given in the proof.

Although lengthy, proving the claimed approximation ratios is not that technically involved. It boils down to proving the error inflicted by computing shortest paths to representatives of sources instead of directly to the sources. For this we rely extensively on the properties of skeleton graphs (given in Lemmas \ref{lem:long-paths-skeleton} and \ref{lem:long-shortest-paths-skeleton}) and the triangle inequality. Subsequently, we leverage the exact knowledge of distances of paths with few hops to turn additive errors into small multiplicative errors.

\subsection{Proof of Theorem \ref{thm:cc-sp-sim}}
\label{sec:cc-sp-sim-proof}

The most straight forward part of the proof of Theorem \ref{thm:cc-sp-sim} is showing the runtime of Algorithm \ref{alg:cc-sp-sim}. Note that in the formulation of this and subsequent lemmas we will assume Algorithm \ref{alg:cc-sp-sim} receives a \CC algorithm $\calA$ as input, which is parametrized according to the specifications given in Theorem \ref{thm:cc-sp-sim}.

\begin{lem}
	\label{lem:cc-sim-runtime}
	\hspace*{-1mm}Given $n^{{x}\gamma}\!$ sources, Algorithm \ref{alg:cc-sp-sim} takes \smash{$\tilO\big(\eta  n^{1-x}\big)$} rounds.
\end{lem}

\begin{proof}
	During the execution of Algorithm \ref{alg:cc-sp-sim} we conduct two explorations in the local network, the longer of which takes $\eta h \in \tilO(\eta  n^{1-x})$ rounds ($h$ is the parameter from Lemma \ref{lem:long-paths-skeleton}). Moreover, we simulate the \CC algorithm $\calA$ on the skeleton $\calS$. A round of simulation takes \smash{$\tilO\big(n^{2{x}-1} \!+ n^{{x}/2}\big)$} rounds in accordance with Corollary \ref{cor:cc-sp-sim}. 
	
	Note that for the chosen \smash{$x = \frac{2}{3 + 2 \delta}$} one exponent dominates the runtime of the simulation, as \smash{$2x\!-\!1 = \tfrac{1-2\delta}{3 + 2 \delta} \leq \tfrac{1}{3 + 2 \delta} = \tfrac{x}{2}$}. Therefore the runtime simplifies to $\tilO\big(n^{{x}/2}\big)$. 	
	Since the number of nodes on which we simulate $\calA$ is only $|V_\calS| = \tilO(n^x)$, running $\calA$ takes $\tilO\big(\eta  n^{x\delta}\big)$ \textit{simulated} rounds, that is $\tilO\big(\eta  n^{x\delta + x/2}\big)$ \textit{actual} rounds. For the value we chose for $x$ the exponent equals 
	$$x \delta \!+\! \tfrac{x}{2} = \tfrac{2\delta}{3+2\delta} \!+\! \tfrac{1}{3+2\delta} = 1 \!-\! \tfrac{2}{3+2 \delta} = 1\!-\!x,$$ thus running $\calA$ takes $\tilO(\eta  n^{1-x})$ rounds in total. 
	
	Finally, we note that since each source in $G$ picks at most one representative in $V_\calS$ that acts as a source in $\calS$, the number of sources $ \Theta\big(|V_\calS|^\gamma\big) = \Theta\big(n^{{x}\gamma}\big)$ in $\calS$ is roughly in line with the allowed number to execute $\calA$ on $\calS$. Specifically, we have at most a constant factor more sources on $\calS$ than $\calA$ allows w.h.p. (due to the randomized sampling of $V_\calS$). This does not matter much as we can repeat $\calA$ a constant number of times with a constant fraction of sources. 
\end{proof}

Let us now turn to the correctness of the distance approximation computed by Algorithm \ref{alg:cc-sp-sim}.

\begin{lem}
	\label{lem:approx-ratio}
	On unweighted graphs with $n^{{x}\gamma}$ sources and $\gamma >0$, Alg.\ $\calB$ (represented by Alg.\ \ref{alg:cc-sp-sim}) computes an \smash{$\big(\alpha \!+\! \tfrac{2}{\eta } \!+\! \tfrac{\beta}{T_\calB}\big)$}-approximation w.h.p., where \smash{$T_\calB \in \tilO\big( \eta n^{1-x}\big)$}. We get a \smash{$\big(1 \!+\! 2\alpha \!+\! \frac{\beta}{T_\calB}\big)$}-approximation on weighted graphs, w.h.p.
\end{lem}

\begin{proof}
	Let $v,s \in V$ be a node and a source, respectively. In case $hop(u,v) \leq \eta  h$, the value $\tild(v,s)$ computed in Equation \eqref{eq:cc-sp-sim} is equal to $d_{\eta  h}(v,s)$ (since $s$ is a source and due to Lemma \ref{lem:long-paths-skeleton}), i.e. the algorithm computes the exact distance. From now on we assume $hop(u,v) > \eta  h$.	
	
	Next we give some definitions and state a few facts that we implicitly use for the subsequent calculations. Let $s\in V$ be a source, let $r_s \in V_\calS$ be the representative of $s$ and let $Q$ be a shortest $v$-$r_s$-path. According to Lemma \ref{lem:long-paths-skeleton} there is a skeleton node $w \in V_\calS$ on $Q$ within $h$ hops of $v$ (possibly $w = r_s$). Note that $v$ learns the distance $d_h(v,w)$ and in fact $d_h(v,z)$ for any $z \in V_\calS$ within $h$ hops during the exploration via the local network in Algorithm \ref{alg:cc-sp-sim}.	
	
	By definition, $\calA$ computes an approximation $\tild(u,r_s)$ with $$d(u,r_s) \leq \tild(u,r_s) \leq \alpha \!\cdot\! d(u,r_s) \!+\! \beta$$ from any skeleton node $u \in V_\calS$ to $r_s$ and after the algorithm has finished, the values $\tild(u,r_s)$ are known by $v$ (c.f. Fact \ref{fct:cc-sp-sim}). Finally, since $s$ must have at least one skeleton node within $h$ hops distance, its representative $r_s \in V_\calS$ (the skeleton closest to $v$) must clearly be within $h$ hops of $s$ as well. 
	The distance label $d_h(r_s,s)$ is made public knowledge, i.e., also known by $v$ (c.f.\ Fact \ref{fct:compute-reps}).
	We have 
	\begin{align*}
	w(Q) & = d(v,r_s) = d_h(v,w) + d(w,r_s)\\
	&= \min_{u\in V_\calS} d_h(v,u) + d(u,r_s)\\
	&\leq \min_{u\in V_\calS} d_h(v,u) + \tild(u,r_s).
	\end{align*}
	Let $P$ be an \textit{actual shortest path} from $v$ to $s$. We know $|P| \geq \eta h$. We show that we do not underestimate $w(P)$
	\begin{align*}
	w(P) & =  d(v,s) \leq d(v,r_s) + d_h(r_s,s) = w(Q) + d_h(r_s,s) \\ & \leq \min_{u\in V_\calS} d_h(v,u) + \tild(u,r_s) + d_h(r_s,s) = \tild(v,s).
	\end{align*}
	Note that the right hand side of the above inequality equals the right hand side of Equation \eqref{eq:cc-sp-sim}, given that $hop(u,v) > \eta h$. Next we upper bound the estimation $\tild(v,s)$ given in Equation \eqref{eq:cc-sp-sim} with a constant multiple of $w(P)$ as follows 
	\begin{align}
	\begin{split}	
	\label{eq:approx-ratio}	
	\tild(v,s) &= \smash{\min_{u\in V_\calS}} d_h(v,u) + \tild(u,r_s) + d_h(r_s,s) \\
	&\leq d_h(v,w) + \tild(w,r_s) + d_h(r_s,s)\\
	&\leq d_h(v,w) + \alpha\!\cdot\!d(w,r_s) \!+\! \beta + d_h(r_s,s).
	\end{split}
	\end{align}
	The first inequality is due to the fact that $d_h(v,w) + \tild(w,r_s)$ is one of the options occurring in the minimum function.
	Assume $G$ is \textit{unweighted}. Then from the above inequality and by exploiting $d(v,s) = hop(v,s) > \eta  h$ we obtain
	\begin{align*}
	\tild(v,s) & \leq d_h(v,w) + \alpha\!\cdot\!d(w,r_s) \!+\! \beta + d_h(r_s,s) 
	\leq \alpha \!\cdot\! d(v,s) + \beta + 2h\\
	&\leq \alpha \!\cdot\! d(v,s) + \tfrac{\beta}{\eta h}d(v,s) + \tfrac{2}{\eta}d(v,s)
	= \big(\alpha \!+\! \tfrac{\beta}{\eta  h} \!+\! \tfrac{2}{\eta }\big)d(v,s)
	\end{align*}	 
	Now assume that $G$ is \textit{weighted}. We continue from Inequality \eqref{eq:approx-ratio}, where we replace $w$ with the first skeleton node $w' \in V_\calS$ on a shortest path from $v$ to $s$ (which also upper bounds the minimum). Note that $hop(v,w') \leq h$ (by Lemma \ref{lem:long-paths-skeleton}). More explanations about the individual steps are given after the following equations.
	\begin{align*}
	\tild(v,s) & \leq d_h(v,w') + \alpha\!\cdot\!d(w',r_s) \!+\! \beta + d_h(r_s,s)\\
	& \leq d_h(v,w') + \alpha\!\cdot\!\big(d(w',s)+ d(s,r_s)\big) \!+\! \beta + d_h(r_s,s)\\
	& \leq \alpha\!\cdot\!\big(d_h(v,w') + d(w',s)\big) + (1 \!+\! \alpha)d_h(r_s,s) + \beta \tag*{\small ($\alpha \geq 1$)}\\
	& = \alpha\!\cdot\!d(v,s) + (1 \!+\! \alpha)d_h(s,r_s) + \beta \tag*{\small (Def. of $w'$)} \\
	& \leq \alpha\!\cdot\!d(v,s) + (1 \!+\! \alpha)d(v,s) + \beta \tag*{\small ($d_h(s,r_s) \leq d(v,s)$)} \\
	&  = (1 \!+\! 2\alpha)\!\cdot\! d(v,s) + \beta \\
	& \leq (1 \!+\! 2\alpha + \beta/\eta  h)\!\cdot\! d(v,s) \tag*{\small ($d(v,s) \geq \eta  h$)}
	\end{align*}
	In the second last inequality we exploit that $r_s \in V_\calS$ minimizes $d_h(s,r_s)$ and therefore $d_h(s,r_s) \leq d(v,s)$ w.h.p., as we show in the following. For a contradiction, presume $d_h(s,r_s) > d(v,s)$. Since $hop(v,s) \geq h$ there is a skeleton node $r'$ on a shortest $v$-$s$ path $P$ within $h$ hops of $s$ w.h.p. due to Lemma \ref{lem:long-paths-skeleton}. But $d_h(s,r') \leq d(v,s) < d_h(s,r_s)$, thus $r'$ would be the representative of $s$.
	
	In the last inequality we use $d(v,s) \geq \eta  h$ for the weighted case. Since the minimum weight of an edge is one, this follows because $P$ has more than $\eta  h$ edges.
	
	Finally, we note that we can do the exploration of the graph via the local network in parallel during the whole runtime $T_\calB \in \tilO\big(\eta n^{1-x}\big)$ of Algorithm \ref{alg:cc-sp-sim}, which gives exact $T_\calB$-hop distances. Thus, we can actually replace occurrences of $\eta h$ in the above approximations with $T_\calB$ (the terms are equal up to $\polylog n$ factors anyway).
\end{proof}

Next we show that the number of sources does not have to be restricted if $\calA$ is an APSP algorithm. The downside of having many sources is an additional \smash{$\tilO\big(\sqrt{k}\big)$} rounds required to make representatives public knowledge.

\begin{lem}
	If $\calA$ solves APSP in \CC we achieve the same parameters as in Lemma \ref{lem:approx-ratio} for an arbitrary number of $k \in [n]$ sources, w.h.p. This adds $\tilO\big(\sqrt{k}\big)$ rounds to the runtime $T_\calB$ of Algorithm \ref{alg:cc-sp-sim}.
\end{lem}

\begin{proof}
	From the point of correctness, the only reason that is prohibiting the use of an arbitrary number of sources is that the number of representatives tagged as sources in $\calS$ by sources in $G$ might be significantly more than $|V_\calS|^\gamma$, i.e.\ the number of sources allowed by $\calA$. If $\calA$ is an APSP algorithm we have no such restriction and may allow an arbitrary number of sources in $G$.
	
	With respect to runtime, each source in $G$ must announce its representative in $\calS$ and the distance to it, in order that correctness is preserved. In case we have $k$ sources this takes $\tilO\big(\sqrt{k}\big)$ time using token dissemination (c.f., Lemma \ref{lem:tokenDissemination}), whereas the rest of the algorithm and its runtime remains unaffected. Therefore, the runtime only increases by \smash{$\tilO\big(\sqrt{k}\big)$} rounds.
\end{proof}	

Finally, we show that the approximation factor improves significantly for SSSP. This is due to the observation that we can always summon the single source into the skeleton without straining helpers in the subsequent token routing protocol too much (each node has to help at most one additional node).

\begin{lem}
	If $\calA$ solves SSSP in the \CC, then $\calB$ gives an $(\alpha\!+\! {\beta}/{ T_\calB})$-approximation on weighted graphs w.h.p.
	%
\end{lem}

\begin{proof}
	Algorithm \ref{alg:compute-skeleton} always adds $s$ to the skeleton if it is a single source. It is not hard to see that the token routing protocol still works if we add a single node to the randomly sampled skeleton $\calS$. Therefore, runtime of Algorithm \ref{alg:cc-sp-sim} remains unaffected. However, the approximation factor improves since the distance to the source does not have to take into account the detour via a representative. 
	
	Specifically, all nodes learn the distance approximations $\tild(u,s)$ computed by $\calA$ for all $u \in V_\calS$. Additionally each node $v \in V$ knows the distance to the first skeleton $w \in V_\calS$ on a shortest path $P$ from $v$ to $s$ w.h.p. (since $w$ must be within $h$ hops due to Lemma \ref{lem:long-paths-skeleton}). Hence in any case all nodes learn at least an $(\alpha,\beta)$-approximation by taking the minimum of $d_h(v,u)+d(u,s)$ for all $u \in V_\calS$ (c.f.\ Equation \eqref{eq:cc-sp-sim} where $r_s = s$).
	
	If the shortest $v$-$s$-path has fewer than $\eta  h$ hops, we can give an exact solution (c.f. Fact \ref{fct:cc-sp-sim} and Equation \eqref{eq:cc-sp-sim}). Otherwise, we know that $d(v,s) \geq \eta  h$ (and even $d(v,s) \geq T_\calB \in \tilO(\eta  h)$ since each node learns $G$ up to depth $T_\calB$ during the whole runtime of $\calB$). Thus, we can turn the additive error into a small multiplicative error
	\begin{align*}
	\tild(v,s) & \leq \alpha \!\cdot\! d(v,s) + \beta \leq \big(\alpha \!+\! \tfrac{\beta}{d(v,s)}\big) \!\cdot\! d(v,s) \\
	&\leq \big(\alpha \!+\! \tfrac{\beta}{\eta  h}\big) \!\cdot\! d(v,s) \leq \big(\alpha \!+\! \tfrac{\beta}{T_\calB}\big) \!\cdot\! d(v,s). \qedhere
	\end{align*}
\end{proof}

\end{longversion}

\subsection{Implied Results for k-SSP and SSSP}

In the following subsection we apply Theorem \ref{thm:cc-sp-sim} on known results for the \CC model \cite{Censor-Hillel2019a,Censor-Hillel2019b}, which will generate a series of results for the \HYBRID model for different numbers of sources, approximation parameters and time complexity.

The first \CC algorithm is from \cite{Censor-Hillel2019a} (Theorem 1.2), which allows for a variable number $n^{\gamma}$ of sources and achieves a runtime of $\tilO(1/\eps)$ for $\gamma = 1/2$ and approximation factor $(1\!+\!\eps)$. Note that the corollary below could be stated more generally for arbitrary $\gamma \geq 1/2$, which we opted not to do for conciseness.

\begin{cor}
	There is an algorithm that approximates the \smash{$n^{1/3}$}-source shortest path problem in the \HYBRID model with running time $\tilO(n^{1/3}/\eps)$ and approximation factor $(1\!+\!\eps)$ on unweighted graphs and $(3\!+\!\eps)$ on weighted graphs, w.h.p.
\end{cor}

\begin{proof}
	Let $\calA$ be the algorithm from \cite{Censor-Hillel2019a} Theorem 1.2. We fix $\gamma = 1/2$. Then $\calA$ has the parameters  $\delta = 0$, $\eta  = 1/\eps$ and  $\alpha = (1\!+\!\eps)$. This implies $x = 2/3$.	
	From Theorem \ref{thm:cc-sp-sim} we obtain a runtime of $\tilO(n^{1/3}/\eps)$ on graphs with $n^{1/3}$ sources. On weighted graphs the approximation factor is \smash{$2\alpha + 1 = 3 + 2\eps \stackrel{\eps':=2\eps}{=} 3+\eps'$}. On unweighted graphs it is \smash{$\alpha + \tfrac{2}{\eta } = \alpha + 2 \eps = 1 + 3 \eps \stackrel{\tilde \eps := 3\eps }{=} 1+\tilde\eps$}.
\end{proof}

The next \CC algorithm stems from \cite{Censor-Hillel2019a} (Theorem 1.1). It solves APSP in $\tilO(1/\eps)$ rounds with approximation parameters $(\alpha, \beta) = (2\!+\!\eps,(1\!+\!\eps)w_{uv})$ for a path from $u$ to $v$, whereas $w_{uv}$ is the heaviest edge on that path. 

\begin{cor}
	There is an algorithm that approximates shortest paths for $k$ sources in the \HYBRID model with running time \smash{$\tilO\big(\frac{n^{1/3}}{\eps} \!+ \sqrt{k}\big)$} and approximation factor $(7\!+\!\eps)$ on weighted graphs and $(2\!+\!\eps)$ on unweighted graphs, w.h.p.
\end{cor}

\begin{proof}
	The defining parameters of the algorithm from \cite{Censor-Hillel2019a} Theorem 1.1 are $\gamma = 1$, $\delta = 0$ and $\eta  = 1/\eps$. This implies $x = 2/3$ and a runtime of $\tilO(n^{1/3}/\eps)$ of the resulting \HYBRID algorithm by Theorem \ref{thm:cc-sp-sim}.		
	The multiplicative error is $\alpha = (2\!+\!\eps)$ and the additive error computed on the skeleton is $\beta = (1\!+\!\eps)w_{uv}$, where $w_{u,v}$ is the heaviest edge on a shortest path from $u$ to $v$. 	
	
	We first consider the case where $G$ is weighted. We observe that $(1\!+\!\eps)w_{uv} \leq (1\!+\!\eps)d(u,v)$, thus we can get rid of the additive error of $\calA$ by simply adding $(1\!+\!\eps)$ to the multiplicative error instead. That is, we can consider $\calA$ as $(3\!+\!2\eps)$-approximation algorithm. Then, by Theorem \ref{thm:cc-sp-sim}, we have a multiplicative error $2(3 + 2\eps) + 1 =$ \smash{$ 7 + 4 \eps \stackrel{\tilde\eps := 4 \eps}{=} 7 + \tilde\eps$}.
	
	Let us consider the unweighted case. We can also bound the additive error by $(1\!+\!\eps)W_{\calS}$ where $W_{\calS}$ is the largest edge weight of the skeleton $\calS$. Since $G$ is unweighted we have $W_\calS \leq h$. Then, by Theorem \ref{thm:cc-sp-sim} we obtain a multiplicative error of 
	\vspace*{-1mm}
	\[
		\alpha + \tfrac{2}{\eta } + \tfrac{(1+\eps)W_\calS}{T_\calB} \!\stackrel{T_\calB \geq \eta h}{\leq}\!\!\!\! \alpha + 2\eps + \tfrac{(1+\eps)h}{\eta h} \leq 2 + 3\eps + \tfrac{1+\eps}{1/\eps} \!\stackrel{\eps \leq 1}{\leq}\! 2 + 5 \eps \stackrel{\eps' := 5 \eps}{=}\!  2 + \eps'\!\!.	\vspace*{-2mm}
	\]
\end{proof}

We turn to a result of \cite{Censor-Hillel2019b}, which offers a more accurate APSP algorithm for the \CC. They solve APSP with multiplicative error $(1\!+\!o(1))$ in \smash{$\tilO(n^{\rho+o(1)})$} rounds, where $\rho$ is the so called \textit{distributed matrix multiplication coefficient} \cite{Censor-Hillel2019b}. The value $\rho \leq 1\!-\! \frac{2}{\omega}$ is algorithmically bounded by the \textit{matrix multiplication coefficient} $\omega$, for which $2 \leq \omega < 2.3728639$ is known \cite{Gall2014}, thus $\rho < 0.15715$. 
It is conjectured by some that $\omega = 2$, raising the possibility of further improvements in the bound of $\rho$. To keep the result brief, we do not express the runtime as a function of $\rho$ and only insert its known upper bound. The algorithm of \cite{Censor-Hillel2019b} allows edge weights up to \smash{$2^{(n^{o(1)})}$} which is $n^c$ if we choose $(\log \log n + \log c)/{\log n}$ for $o(1)$. 
This and the previous two corollaries constitute the proof of Theorem \ref{thm:kSSP}.

\begin{cor}
	There is a \HYBRID algorithm for the $k$-SSP problem with runtime \smash{$\tilO\big(n^{0.397} \!+ \sqrt{k}\!\big)$} and approx.\ factor $(3\!+\!o(1))$ on weighted graphs and runtime \smash{$\tilO\big(n^{0.397}/\eps + \sqrt{k}\!\big)$} and approx.\ factor $(1\!+\!\eps)$ on unweighted graphs w.h.p.
\end{cor}

\begin{proof}
	We characterize the above algorithm with $\delta = \rho < 0.15715$, $\eta = 0$ and $\alpha = 1 \p o(1) \leq 1 \p \varepsilon$ (for sufficiently large $n$) and $\beta = 0$. We (somewhat artificially but w.l.o.g.) increase $\eta = 1/\eps$ for the unweighted case (to do a more extensive local exploration). Then the runtime is \smash{$\tilO\big(n^{0.397}/\eps + \sqrt k\big)$} in the unweighted and \smash{$\tilO\big(n^{0.397} \!+ \sqrt k\big)$} in the weighted case.
	On unweighted graphs the approximation is \smash{$\alpha \p \tfrac{2}{\eta} \leq 1 \p 3 \eps \!\stackrel{\eps' := 3 \eps}{=}\! 1 \p \eps'$}. On weighted graphs we have \smash{$2 \alpha \p 1 = 3 \p o(1)$}.
\end{proof}

The last corollary is derived from \cite{Censor-Hillel2019a} Theorem 5.2, stating that SSSP can be solved \textit{exactly} in the \CC model in $\tilO(n^{1/6})$ rounds. As we have shown in the proof of Theorem \ref{thm:cc-sp-sim}, we can conserve this approximation ratio (or rather lack thereof) in case of a single source, which gives us an exact algorithm for the \HYBRID model. 

\begin{cor}[Theorem \ref{thm:SSSP}]
	There is a \HYBRID algorithm that solves SSSP exactly in $\tilO(n^{2/5})$ rounds, w.h.p.
\end{cor}

\begin{proof}
	The SSSP algorithm from \cite{Censor-Hillel2019a} can be parametrized by $\gamma = 0$, $\delta = 1/6$, $\eta = 1$ and $\alpha = \beta = 0$.
\end{proof}

	\section{Simulation Results for Diameter}
\label{sec:cc-diam-sim}

In the distributed diameter problem all nodes must learn the parameter $D = D(G)$, which equals the maximum hop-distance $hop(u,v)$ for any pair $u,v\in V$. There are fast \CC algorithms known for the \textit{weighted} diameter problem and the \textit{weighted} APSP problem (which obviously solves the weighted diameter problem). We aim to apply the same technique of simulating the \CC on a skeleton graph of a \HYBRID network to the diameter problem.

The core idea is relatively simple. If $D$ is large, then the diameter of the skeleton will give a good approximation of the diameter of the whole graph. Therefore, computing a good approximation on the skeleton using a known \CC algorithm will give a decent approximation for the \HYBRID model. If, on the other hand, $D$ is relatively small we can use the local network to compute $D$ exactly.
The following theorem formalizes the result in a more general way.

\begin{thm}
	\label{thm:cc-diam-sim}
	Let $\calA$ be an $(\alpha,\beta)$-approximation algorithm for the \CC model for the \emph{weighted} diameter with runtime \smash{$T_\calA \in \tilO(\eta n^\delta)$} for const.\ $\delta \geq 0, \eta \geq 1$. Let ${x} \!:=\! \frac{2}{3 + 2 \delta}$. 	
	Then there is an algorithm $\calB$ for the \HYBRID model with runtime \smash{$T_\calB \in \tilO\big(\eta n^{1-{x}}\big)$} that gives an \smash{$\big(\alpha \!+\! \tfrac{2}{\eta} \!+\! \tfrac{\beta}{T_\calB}\big)$}-approximation of the \emph{unweighted} diameter $D$ w.h.p.
\end{thm}

\shrt{The diameter algorithm as well as the full proof of Theorem \ref{thm:cc-diam-sim} are given in the full paper \cite{KS20}. We provide a brief summary of both.}\lng{We give a brief summary of the algorithm and its proof before going into the details.} As in Section \ref{sec:cc-sp-sim}, we first compute a skeleton graph $\calS =(V_\calS, E_\calS)$ of appropriate size $\Theta(n^{x})$ (where \smash{${x} := \frac{2}{3 + 2 \delta}$} balances the trade-off between the runtime of the simulation and the runtime for the local exploration of $G$). Second, we simulate the \CC algorithm $\calA$ on $\calS$, which gives each node in $V_\calS$ an approximation $\tilD(\calS)$ of $D(\calS)$. 

Third, we let nodes exchange information via the local model for $\eta h \p 1$ rounds (where \smash{$h \in \tilO(n^{1-x})$} is the maximum length of skeleton edges). This achieves two things. (I) it floods $\tilD(\calS)$ to all nodes (recall that each node has a skeleton within $h$ hops w.h.p.). (II) each node learns its largest hop distance $h_v := \max_{w \in N_v} \leq \eta h \!+\! 1$ to any node $w$ in $v$'s $(\eta h \!+\! 1)$-hop neighborhood $N_v$. Fourth, we apply a standard technique from \cite{augustine19_NCC} (given in Lemma \ref{lem:aggregationProtocol}) to aggregate the maximum value out of $n$ values $h_u, u \in V$ (one held by each node) and to announce the biggest value \smash{$\widehat h = \max_{u \in V}h_u$} to all nodes in \smash{$\tilO(1)$} rounds. Finally, all nodes can compute an approximation $\tilD$ of $D(G)$ by setting \smash{$\tilD := \widehat h$} if \smash{$\widehat h \leq \eta h$ and $\tilD := \tilD(\calS)+2h$} else.

For the correctness, the important observation is that the diameter of the skeleton can be at most $2h$ (the maximum length of two skeleton edges) shorter than the diameter of $G$ and we make up for that by adding $2h$ to the estimation $\tilD(\calS)$. This gives us an additive error of at most $2h \p \beta$, which does not impact our approximation much if we turn it into a multiplicative error, using that we know $D$ exactly if $D \leq \eta h$. We move on to the application of Theorem \ref{thm:cc-diam-sim}.

\begin{longversion}

\subsection{Diameter Simulation - Algorithm and Proof}
\label{apx:cc-diam-sim}

We use the same basic approach of simulating the \CC model on a skeleton, thus the following algorithm is quite similar to the one of the previous section and reuses the according sub-procedures.

\begin{algorithm}[H]
	\caption{\texttt{Diam-Simulation($\calA, \delta, \eta$)} \Comment{\textit{$\delta, \eta, \calA$ as in Th.\ \ref{thm:cc-diam-sim}}}}
	\label{alg:cc-diam-sim}
	\begin{algorithmic}
		\State ${x} \gets \frac{2}{3 + 2 \delta}$ \Comment{\textit{${x}$ optimizes overall runtime}}
		\State \texttt{Compute-Skeleton($\bot,{x}$)} \Comment{\textit{compute $\calS\!=\!(V_\calS,\!E_\calS),|V_\calS| \!\in\! \tilO(n^{x}\!)$}}
		\State \texttt{Clique-Simulation($\calA,x$)} \Comment{\textit{run $\calA$ on $\calS$}}
		\For{$\eta  h + 1$ rounds} \Comment{\textit{$v$ learns approx.\ $\tilD(\calS)$ of $D(\calS)$}}
		\State $v$ forwards diameter approximations via local edges
		\EndFor
		\State $h_v \gets \max_{u \in N_{\eta h + 1}(v)} hop(v,u)$ \Comment{\textit{largest hop-dist.\ $v$ sees locally}}
		\State $v$ participates in aggregation for $\widehat h = \max_{u \in V}h_u$ 
		\State $v$ computes approx.\ $\tilD$ of $D(G)$ with Equation \eqref{eq:cc-diam-sim}
	\end{algorithmic}
\end{algorithm}

\begin{fact}
	\label{fct:cc-diam-sim} 
	With the same arguments as in Lemma \ref{lem:cc-sim-runtime} the runtime of Algorithm \ref{alg:cc-diam-sim} is $\tilO(\eta n^{1- x})$.	
	Algorithm \ref{alg:cc-diam-sim} first computes a skeleton with the properties outlined in Fact \ref{fct:compute-skeleton}. The simulation of a diameter algorithm $\calA$ makes an $(\alpha,\beta)$-approximation $\tilD(\calS)$ of $D(\calS)$ known to all skeleton nodes. The subsequent local exploration of graph $G$ for $\eta h \!+\! 1$ rounds accomplishes two things. First, it spreads \smash{$\tilD(\calS)$} to all nodes (recall that each node has a skeleton within $h$ hops w.h.p.). Second, each node learns the largest hop distance $h_v \leq \eta h \!+\! 1$ to any node $u$ in $v$'s $(\eta h \!+\! 1)$-neighborhood. We apply the aggregation scheme of \cite{augustine19_NCC} (c.f., Lemma \ref{lem:aggregationProtocol}) to aggregate the maximum value out of $n$ values $h_u, u \in V$ (one held by each node) and announce the biggest value \smash{$\widehat h = \max_{u \in V}h_u$} to all nodes in $\tilO(1)$ rounds. Finally, each node computes an approximation of $\tilD$ of $D$ as follows	
	\begin{equation}
	\label{eq:cc-diam-sim}
	\tilD := 
	\begin{cases}
	\widehat h, & \text{if } \widehat h \leq \eta h\\
	\tilD(\calS)+2h, &\text{else}
	\end{cases}
	\end{equation}
\end{fact}

The runtime is already accounted for in Fact \ref{fct:cc-diam-sim}. The following lemma completes the proof of Theorem \ref{thm:cc-diam-sim}.

\begin{lem}
	Algorithm \ref{alg:cc-diam-sim} computes an \smash{$\big(\alpha \!+\! \tfrac{2}{\eta} \!+\! \tfrac{\beta}{T_\calB}\big)$}-approximation of the diameter $D$ of an unweighted graph $G$ w.h.p.
\end{lem}

\begin{proof}
	As summarized in Fact \ref{fct:cc-diam-sim}, each node already knows an \smash{$(\alpha,\beta)$}-approximation \smash{$\tilD(\calS)$} of the diameter of $\calS$ and the maximum hop distance \smash{$\widehat h$} that any node sees in its $(\eta h \!+\! 1)$-neighborhood. First, assume $D \leq \eta h$. By definition there is a pair of nodes $u,v \in V$ with $hop(u,v) = D$ and therefore \smash{$\tilD=\widehat h = D$} by Equation \eqref{eq:cc-diam-sim} and the definition of $\widehat{h}$. We would like to point that if \smash{$\widehat h \leq \eta h$}, then all nodes can be sure that $D \leq \eta h$ and thus \smash{$\widehat{h} =D$}. Because if $D$ would be larger, then \smash{$\widehat{h} = \eta h \!+\! 1$} and the nodes use the diameter approximation $\tilD = \tilD(\calS) + 2h$ instead, by Equation \eqref{eq:cc-diam-sim}.
	
	Consequently, we now have to look at the case $D > \eta h$. We analyze how good the approximation $\tilD = \tilD(\calS)+2h$ is for $G$. By definition we know $D(\calS) \leq \tilD(\calS) \leq \alpha D(\calS) + \beta$. Next, we relate the diameter of the skeleton $\calS$ to the diameter of $G$.	
	Let $P$ be a shortest path in $G$ with $|P| = D$. By Lemma \ref{lem:long-paths-skeleton} there are two skeleton nodes $u,v \in V_\calS$ within the first $h$ hops and the last $h$ hops of $P$ w.h.p. By Lemma \ref{lem:long-shortest-paths-skeleton} we have $d_\calS(u,v) = d(u,v) = hop(u,v)$. Consequently, $D(\calS) \geq d(u,v) \geq |P|-2h = D - 2h$. As the skeleton graph preserves distances (Lemma \ref{lem:long-shortest-paths-skeleton}) we have $D(\calS) \leq D$. Thus
	\[
	D \leq D(\calS) + 2h \leq \tilD(\calS) + 2h = \tilD \leq \alpha D(\calS) + \beta + 2h \leq \alpha D + \beta + 2h.
	\]
	Note that the computed approximation $\tilD$ in the middle is squeezed from above and below by expressions in $D$. Since we know that we have computed the diameter exactly if it was smaller than $\eta h$, we can apply the same trick as in the previous section and turn the additive approximation error into a smaller multiplicative one.
	\[
	\tilD \leq \alpha D + \beta + 2h \leq \big(\alpha \!+\! \tfrac{\beta}{D} \!+\! \tfrac{2h}{D}\big) D \leq \big(\alpha \!+\! \tfrac{\beta}{\eta h} \!+\! \tfrac{2}{\eta}\big)D
	\]
	Finally we point out that we can do the local exploration for $T_\calB+1$ rounds without changing the asymptotic runtime of the algorithm (since $T_B \in \tilT(\eta h)$ anyway). Hence we can get rid of the parameter $h$ in our expression: \smash{$\tilD \leq \big(\alpha \!+\! \tfrac{\beta}{T_\calB} \!+\! \tfrac{2}{\eta}\big)D$}.
\end{proof}

\end{longversion}

\subsection{Implied Results for Computing Diameter}

First, we apply a dedicated diameter \CC algorithm from \cite{Censor-Hillel2019a}. It achieves a $\big(\tfrac{3}{2} \!+\! \eps, W\big)$-approximation in $\tilO\big(\tfrac{1}{\eps}\big)$ rounds, where $W$ is the maximum edge weight. (Note that this and the subsequent corollary prove Theorem \ref{thm:diameter}.)

\begin{cor}
	There is a \HYBRID algorithm that computes a \smash{$(\tfrac{3}{2}\!+\!\eps)$}-approximation of $D(G)$ in \smash{$\tilO\big(\tfrac{n^{1/3}}{\eps}\big)$} rounds w.h.p.
\end{cor}

\begin{proof} The algorithm of \cite{Censor-Hillel2019a} can be characterized with $\delta = 0, \eta = 1/\eps$ and $(\alpha,\beta) = \big(\tfrac{3}{2} \!+\! \eps, W\big)$. The runtime follows immediately from Theorem \ref{thm:cc-diam-sim}. Since $G$ is unweighted, the maximum edge weight in the skeleton $\calS$ is $h$. Thus, the approx.\ factor is
	\[
		\alpha + \tfrac{2}{\eta} + \tfrac{\beta}{T_\calB} \leq \tfrac{3}{2} + \eps + 2 \eps + \tfrac{W}{\eta h} \leq  \tfrac{3}{2} + 3 \eps + \tfrac{h}{h/\eps} = \tfrac{3}{2} + 4 \eps \stackrel{\eps' := 4 \eps}{=} \tfrac{3}{2} + \eps'.
		\vspace*{-2mm}
	\]
\end{proof}

Another noteworthy result can be deduced from the APSP algorithm of \cite{Censor-Hillel2019b} that we used before for shortest paths. It has runtime \smash{$\tilO\big(n^{\rho+o(1)}\big) \subseteq  \tilO\big(n^{0.15715}\big)$} (for $n$ large enough) and a multiplicative error $\big(1\!+\!o(1)\big)$ (recall that $\rho$ is the distributed matrix multiplication coefficient with \smash{$\rho \leq 1\!-\! \frac{2}{\omega} < 0.15715$}). It allows for edge weights up to \smash{$2^{(n^{o(1)})}$} which is $n$ if we choose $\log \log n / \log n$ for the $o(1)$. To keep the result concise, we omit expressing it as a function of $\rho$ and insert the known upper bound instead. 

\begin{cor}
	There is a \HYBRID algorithm that computes a $(1\!+\!\eps)$-approximation of $D(G)$ in \smash{$\tilO\big(\tfrac{n^{0.397}}{\eps}\big)$} rounds w.h.p.
\end{cor}

\begin{proof}
	Clearly, solving APSP on the skeleton also solves the diameter problem, hence Theorem \ref{thm:cc-diam-sim} is applicable. We characterize the above algorithm with parameters $\delta = 0.15715, \eta = 1/\eps$ and $\alpha = 1 + o(1) \leq 1 + \varepsilon$ (for sufficiently large $n$) and $\beta = 0$. Then the runtime is $\tilO\big(n^{0.397}/\eps\big)$ and the approximation parameter is \smash{$\alpha + \tfrac{2}{\eta} + \tfrac{\beta}{T_\calB} \leq 1 + 3 \eps \stackrel{\eps' := 3 \eps}{=} 1+ \eps'$}.
\end{proof}

	\begin{shortversion}

\section{Lower Bounds}
\label{sec:lower-bounds}

We complement our algorithmic upper bounds with corresponding lower bounds. We only have room to summarize our strategy, whereas the details are given in the full paper \cite{KS20}. The first bound stems from the observation that the lower bound for APSP of \cite{Augustine2020} can be adapted, which proves that computing an approximation of the $k$-source shortest paths problem takes \smash{$\tilOm\big(\sqrt k\big)$} rounds, even on unweighted graphs and for large approximation factors up to some \smash{$\alpha \in \Theta\big(n/\sqrt k\big)$}. This implies that our constant approximation, $k$-source shortest path algorithms with runtime \smash{$\tilO\big( \sqrt{k} \big)$} (for $k \geq n^{1/3}$ and $k \geq n^{0.397}$ respectively) are tight up to $\polylog n$ factors (c.f.\ Theorem \ref{thm:kSSP} and the subsequent discussion).

The idea is to construct an path that has a dedicated node $b$ at one end. Moreover, we have two sets of sources $S_1,S_2$ of size $k$ each, attached (with an edge) to the path at two dedicated nodes $v_1,v_2$ on the path. Node $v_2$ is at the opposite end of $b$ and $v_1$ is at distance $L \in \tilT\big(\sqrt{k}\big)$ from $b$. Now an adversary randomly ``shuffles'' the sources in $S_1,S_2$. We consider the state of $S_1,S_2$ as a random variable $X$. Then $X$ has entropy $H(X) = \tilOm(k)$ bits. If $b$ does not know for a single source $s$ whether $s \in S_1$ or $s \in S_2$ then $b$ can not give a satisfying answer to the $k$-SSP problem.
Hence, $b$ has to learn the state of $X$. If given only $L\m 1$ rounds then using only the local network is out of the question since the information of $X$ can not reach $b$. In fact, all information about $X$ that can reach $b$ in the given time must at some point be transmitted to a node within at most $L\m 1$ hops of $b$ via a global message. But the combined receive capacity of all those nodes is at most $L \in \tilT\big(\sqrt{k}\big)$ bits per round. Hence it takes $\tilOm\big(\sqrt{k}\big)$ rounds to transmit all $H(X) = \tilOm(k)$ bits.

\end{shortversion}

\begin{longversion}

\section{Lower Bounds for $k$-SSP}
\label{sec:k-src-lb}

Since the lower bound of $\tilOm\big(\sqrt k\big)$ for $k$-source shortest paths is only a slight modification of the lower bound of $\tilOm\big(\sqrt n\big)$ for APSP given in \cite{Augustine2020}, we would like to keep this subsection brief by referring to \cite{Augustine2020} for technical details. In order to keep this paper reasonably self-contained, we give the definition of the according worst case graph and the core idea behind the lower bound. Thus, knowledge of \cite{Augustine2020} is not strictly required to obtain an understanding of the proof.

\begin{figure}[h]
	\centering
	\includegraphics[width=0.95\linewidth]{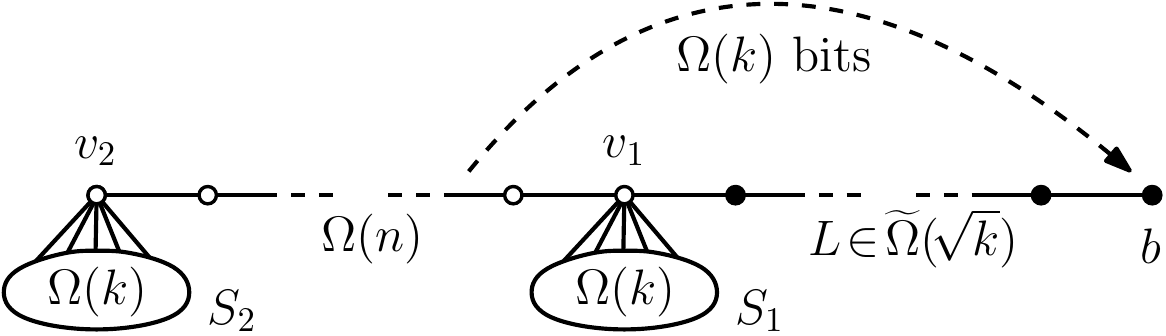}
	\caption{Worst case graph for $k$-SSP.}
	\label{fig:lb-k-source}
\end{figure}

The worst case graph that we use (depicted in Figure \ref{fig:lb-k-source}) is unweighted. It consists of a path of length $\Omega(n)$ with a dedicated node $b$ at one end. Additionally, there are two sets of sources $S_1,S_2$ attached to the path, one to a node $v_1$ at distance \smash{$L \in \tilOm\big(\sqrt k\big)$} from $b$ the other to a node $v_2$ at the opposite end of the path from $b$. Attached means that each of the sources in the sets $S_1,S_2$ has an edge to the respective node $v_1,v_2$.

The first insight is the following. Given a set $S$ of sources of size $|S| = k$ then we randomly assign $|S|/2$ nodes to each set $S_1$ and $S_2$. This creates a random state which $b$ must necessarily learn in order to learn the distance to each node. Considered as a random variable, the state of $S$ carries a Shannon entropy of $\tilOm\big(k\big)$ bits that must somehow be transmitted to $b$ since in order that $b$ can solve the $k$-SSP problem. Note that the difference to \cite{Augustine2020} is $k=n$.

The second insight is that given some graph $G$ that consists of a subgraph $G'$ and a path with $L$ hops attached to $G'$ and a random variable $X$ with entropy of $\Omega(L^2)$ bits that is initially only known to nodes in $G'$, then it takes $\tilOm(L)$ rounds until the node $b$ at the end of the path learns all bits. The intuition behind this is as follows.

Assume we are given at most $L \m 1$ rounds. Then any information that reaches $b$ from $G'$ could not have relied only local edges, since the path has $L$ hops. Therefore, any information that ever reaches $b$ must at some point be transmitted to a node on the path via a global edge. However, the combined capacity of all nodes on the path for receiving information is $\bigO\big(L\log^2 n\big)$ bits. Hence this takes at least $\tilOm(L)$ rounds. In \cite{Augustine2020} this is formalized in the technical Lemma 4.4.

Combining the two insights above already leads to a lower bound of \smash{$\tilOm(L) = \tilOm\big(\sqrt k\big)$} for solving the $k$-source shortest path problem exactly. The third insight concerns the approximation ratio. Assume $b$ does not know for a single node $u$ whether $u \!\in\! S_1$ or $u \!\in\! S_2$. Then $b$ must assume $u \!\in\! S_2$ in order to give a valid approximation (recall that approximations must be lower bounded by the true distance).
But if in reality $u \!\in\! S_1$, then $b$'s distance approximation is off by a factor \smash{$\alpha' \!\in\! \Theta\big(n/\,\sqrt{k}\big)$}. This means that, given slightly less time than required above, $b$ can not even give an $\alpha'$-approximation. Hence, computing an $\alpha$-approximation for any $\alpha \!\leq\! \alpha'$ takes \smash{$\tilOm\big(\sqrt k\big)$} rounds.

From the above insights we conclude Theorem \ref{thm:k-src-lb} (which can be seen as corollary of its counterpart in \cite{Augustine2020}).


\end{longversion}

\lng{\section{Lower Bound for Diameter}
\label{sec:diam-lb}}

The second lower bound is for the diameter problem. We demonstrate that classical proofs of lower bounds for the \CONGEST model based on 2-party set disjointness problems, can be adapted for lower bounds in the \HYBRID model. \shrt{In the full article \cite{KS20} we provide some technical insight how to achieve this for the diameter problem, which may be of use for repeating this procedure for other graph problems based on 2-party set disjointness.}
\lng{The $2$-party set disjointness problem asks whether for two given bit sequences \smash{$a,b \in \{0,1\}^{K}$}, there is an index $i \in [K]$, such that $a_i = b_i = 1$. If two parties (Alice and Bob) are given one of the bit sequences $a,b$ each, then they must exchange $\Omega(K)$ bits to solve this problem \cite{kalyanasundaram92,razborov92}. 
	
To obtain lower bounds in models of distributed computing one has to create a problem instance that solves the 2-party set disjointness problem by solving the problem instance (i.e., provide an appropriate reduction). Additionally, the problem instance must have an information bottleneck between the part of the network that represents Alice and the part that represents Bob. As communication in \HYBRID is conceptually different from communication in \CONGEST, we need to slightly alter the existing worst case graphs and provide additional technical lemmas to demonstrate a corresponding information bottleneck for the \HYBRID model. The technical tools to adapt existing proofs using this framework to the \HYBRID model may also be of use for repeating the procedure for other graph problems based on the same proof strategy.}
The worst case graph in our proof is an adaptation of the one by \cite{holzer_et_al16} for the \CONGEST model and we use a similar simulation argument as in \cite{Sarma2012}. Specifically, we show that \smash{$\tilOm\big(n^{1/3}\big)$} rounds are required to compute the unweighted diameter exactly and to compute a $(2\m\eps)$-approximation of the weighted diameter. The lower bounds hold for randomized protocols, even if both players have access to public randomness.

\begin{longversion}


More specifically, we show how to construct a family of graphs, which for some parameter $k$ encodes set disjointness instances over a universe of size $k^2$ and such that for a given set disjointness input \smash{$a,b \in \{0,1\}^{k^2}$}, the graph has a slightly higher diameter if and only if $a$ and $b$ are disjoint. We will show that if we give $a$ to Alice and $b$ to Bob, Alice and Bob can together efficiently simulate a given \HYBRID model algorithm for computing the diameter of the graph defined by $a$ and $b$. Because by the end of the simulation Alice and Bob know the diameter of the graph (with sufficiently large probability), they have solved set disjointness and they must therefore have exchanged $\Omega(k^2)$ bits. We will see that the parts of the graph simulated by Alice and Bob respectively, can exchange only $\tilO(n)$ bits in the simulated \HYBRID model protocol, resulting in a lower bound of $\tilOm \big(k^2/n\big)$ for the diameter problem (we are able to choose $k \in \Theta(n^{2/3})$).

We use the following family of graphs \smash{$\Gamma^{a,b}_{k,\ell,W} = (V,E)$} to show a bound of \smash{$\tilOm(n^{1/3})$} rounds to compute the diameter in the \HYBRID model. It is based on the construction of \cite{holzer_et_al16} and an example is depicted in Figure \ref{fig:diameter_graph}. We define \smash{$\Gamma^{a,b}_{k,\ell,W}$} as follows.
We have four sets of nodes $V_1,V_2,U_1,U_2$ each of size $k$. These node sets form \textit{cliques} with edges of weight $W$. Then the node sets $V_i$ and $U_i$ ($i=1,2$) are perfectly matched with $\ell$-hop-paths of edges of weight $1$. We further have two nodes $\hat v$ and $\hat u$, connecting to all nodes in $V_1 \cup V_2$ and in $U_1 \cup U_2$, respectively, with edges of weight $W$. The nodes $\hat v$, $\hat u$ are connected by a path of length and hop-length $\ell$.

\begin{figure}
	\centering
	\includegraphics[width=0.8\linewidth]{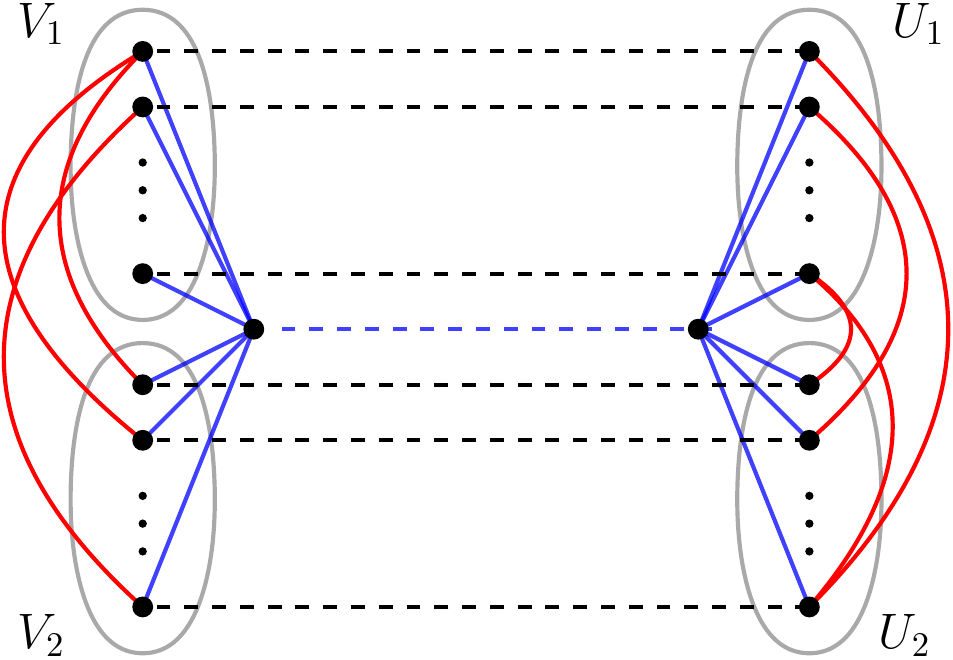}
	\caption{Graph \smash{$\Gamma^{a,b}_{k,\ell,W}$} with $|V_i| \!=\! |U_i| \!=\! k$. The $V_i, U_i$ form internal cliques (gray oval shapes) with edges of weight $W$. Inputs \smash{$a, b \in \{0,1\}^{k^2}\!\!$} are encoded via red edges. Dashed lines are paths with $\ell$ hops. Bold lines are edges of weight $W$. The blue network guarantees $D \leq 2W + \ell$. It is $D = W \!+\! \ell$, if and only if $a,b$ are disjoint.}
	\label{fig:diameter_graph}
\end{figure}

Let \smash{$a,b \in \{0,1\}^{k^2}$} be two input sequences of the set disjointness problem. We represent $a$ and $b$ in \smash{$\Gamma^{a,b}_{k,\ell,W}$} with edges (c.f., Figure \ref{fig:diameter_graph}, red lines). More specifically, we first map each bit $a_i$ in the sequence $a$ to a pair in $p_i \in V_1 \times V_2$ and each bit $b_j$ in $b$ to a pair in $q_j \in U_1 \times U_2$. Since \smash{$|V_1 \times V_2| = |U_1 \times U_2| = k^2$}, we have a pair for each bit. 

Additionally, we assign bits to pairs such that the order of bits in the sequence is consistent with the matching between $V_1,U_1$ and $V_2, U_2$. If for some $i \in [k^2]$ the bit $a_i$ is mapped to $p_i = (v_{1},v_{2})$ and $b_i$ to $q_i = (u_{1},u_{2})$ then $v_{1}, u_{1}$ and $v_{2}, u_{2}$ are  matched with paths of length $\ell$.
Finally, we add an actual edge between a pair $p_i = (v_{1},v_{2})$ or $q_j = (u_{1},u_{2})$ if and only if the corresponding bit $a_i$ or $b_j$ in the sequence is 0 (c.f., Figure \ref{fig:diameter_graph}, red lines).

The main property of the graph \smash{$\Gamma^{a,b}_{k,\ell,W}$} is that it has low diameter if and only if $a$ and $b$ are disjoint. This is shown by the following lemma. We first consider the weighted case, where we are interested in large weights $W > \ell$.

\begin{lem}
	\label{lem:reduction_diam_weighted}
	Suppose $W>\ell$. Then sequences \smash{$a, b \in \{0,1\}^{k^2}\!\!$} are disjoint, if and only if the diameter of \smash{$\Gamma^{a,b}_{k,\ell,W}$} is at most $W \!+\! 2\ell$. Else, the diameter is at least $2W \!+\! \ell$.
\end{lem}

\begin{proof}
	Considering the graph \textit{without} the edges added through $a,b$ (everything except red edges in Figure \ref{fig:diameter_graph}), the distance between any pair of nodes is at most $2W + \ell$. We believe that the respective cases are readily verifiable by the reader (the most noteworthy case is for two nodes on different paths, represented by the dashed lines in Figure \ref{fig:diameter_graph}).
	
	Next, we show that the distance between any pair of nodes is at most $W \!+\! 2\ell$ if $a,b$ are disjoint.
	Regardless of $a,b$, this is true for any node on the path from $\hat u$ to $\hat v$ (including those two) and any other node (c.f., blue path in Figure \ref{fig:diameter_graph}). The distance between such a pair is in fact at most $W \!+\!\ell$. 
	
	This is also true, independently from $a,b$, for any pair of nodes only from either the \textit{top half} of the graph (defined as $V_1 \cup U_1$ and nodes on the matching paths of length $\ell$ running between $V_1$ and $U_1$) or the \textit{bottom half} of the graph (defined analogously, c.f.\ Figure \ref{fig:diameter_graph}). More specifically, the distance between such a pair is at most $W \!+\!\ell$ due to the fact that each of the sets $V_1,V_2,U_1,U_2$ forms an internal clique with edges of weight $W$.
	
	If $a,b$ are disjoint, then (and only then) we can also bound the distance by $W \!+\! 2 \ell$ between two nodes from opposing parts of the graph (one node in the top, the other in the bottom half). Given that $a,b$ are disjoint, for each $i \in [k^2]$ we have $a_i = 0$ or $b_i = 0$ (or both) in the input sequences $a,b$. That means that for any two pairs $p_i = (v_1,v_2)$ and $q_i = (u_1,u_2)$ we have $\{v_1,v_2\} \in E$ or $\{u_1,u_2\} \in E$ by construction.
	
	Let \smash{$w^\uparrow,w^\downarrow$} be two nodes in the top and the bottom half (as defined above) respectively. Let $v_1 \in V_1,u_1 \in U_2$ be the two nodes with a matching path that contains \smash{$w^\uparrow$} (\smash{$w^\uparrow$} may be equal to $v_1$ or $u_1$), and let $v_2\in V_2,u_2\in U_2$ be the analogous nodes for \smash{$w^\downarrow$}. Since $\{v_1,v_2\} \in E$ or $\{u_1,u_2\} \in E$ the distance between \smash{$w^\uparrow$} and \smash{$w^\downarrow$} can be at most $W \!+\! 2 \ell$.
	
	So far we have shown that all possible combinations of two nodes have distance at most $W \!+\! 2 \ell$ if $a,b$ are \textit{disjoint}. It remains to be shown that this is \textit{not} the case 
	if $a,b$ are \textit{not} disjoint. Then we have nodes $v_1 \in V_1, u_1 \in U_1, v_2 \in V_2, u_2 \in U_2$ such that $\{v_1,v_2\} \notin E$ and  $\{u_1,u_2\} \notin E$. It is not hard to see that in order to get from $v_1$ to $u_2$ one must use at least two edges of weight $W$ and at least one path of total weight $\ell$. Hence the diameter of \smash{$\Gamma^{a,b}_{k,\ell,W}$} is at least $2W \!+\! \ell$.
\end{proof}

In the unweighted case $W=1$ we have to proceed even more carefully with our case distinction, specifically for nodes on different paths. However, we can reuse parts of the proof of the previous lemma. We show the following.

\begin{lem}
	\label{lem:reduction_diam_unweighted}
	Suppose $W=1$. Then \smash{$a, b \in \{0,1\}^{k^2}\!\!$} are disjoint, if and only if the diameter of  \smash{$\,\Gamma^{a,b}_{k,\ell,1}$} is $\ell \!+\! 1$. Else, the diameter is $\ell \!+\! 2$.
\end{lem}

\begin{proof}
	As in the previous proof the maximum distance among any pair of nodes is at most $\ell + 2W = \ell + 2$. For nodes on the path from $\hat v$ to $\hat u$ (blue path) we have shown a distance of at most $\ell \!+\! W = \ell \!+\! 1$. The same (distance at most $\ell \!+\! W = \ell \!+\! 1$) is true for pairs of nodes both either in the top half or the bottom half (as defined in the previous proof).
	
	A somewhat different approach is needed to show a distance of at most $\ell \!+\! 1$ between nodes $w^{\uparrow}$ and $w^{\downarrow}$ in the top and the bottom half respectively, if $a,b$ are disjoint.
	Like before let $v_1 \in V_1,u_1 \in U_2$ be the two nodes with a matching path that contains \smash{$w^\uparrow$} and let \smash{$L^{\uparrow}$} be the according path. Let $v_2\in V_2,u_2\in U_2$ and \smash{$L^{\downarrow}$} be the analogous objects for \smash{$w^\downarrow$}. Let $v_1 = x_1, \ldots, x_\ell = u_1$ and $v_2 = y_1, \ldots, y_\ell = u_2$ be the nodes on the paths \smash{$L^{\uparrow}$} and \smash{$L^{\downarrow}$} ordered from left to right. Let $i,j \in [\ell]$ such that \smash{$x_i = w^\uparrow$} and \smash{$y_j = w^\downarrow$}. Presume $i\p j \leq \ell\!-\!1$ then 
	\begin{align*}
	d(w^\uparrow,w^\downarrow) &= d(x_i,y_j) \leq d(x_i,x_1) + d(x_1,y_1) + d(y_1,y_j) \\ &= i  + 2  + j \leq  \ell - 1.
	\end{align*}
	Presume $i+j \geq \ell\!+\!1$. This implies $(\ell \m i) + (\ell \m j)  \leq \ell\!-\!1$. Then almost analogously
	\begin{align*}
	d(w^\uparrow,w^\downarrow) &= d(x_i,y_j) \leq d(x_i,x_\ell) + d(x_\ell,y_\ell) + d(y_\ell,y_j) \\ &= (\ell \m i) + 2 + (\ell \m j) \leq  \ell - 1.
	\end{align*}
	Finally presume $i\p j = \ell$. In this case we actually require the set disjointness. This means we have $\{x_1,y_1\} \in E$ or $\{x_\ell,y_\ell\} \in E$, making one of the distances $d(x_1,y_1)$ or $d(x_\ell,y_\ell)$ equal to $W=1$. Assume $\{x_1,y_1\} \in E$. Then
	\begin{align*}
	d(w^\uparrow,w^\downarrow) & = d(x_i,y_j) \leq d(x_i,x_1) + d(x_1,y_1) + d(y_1,y_j) \\ & = i  + 1  + j =  \ell + 1.
	\end{align*}
	The case $\{x_\ell,y_\ell\} \in E$ is analogous. This covers all possible pairs of nodes, i.e., the distance between any such pair is at most $\ell \p 1$ if $a,b$ \textit{are disjoint}. As in the previous proof, it remains to argue that there is a pair of nodes with distance strictly bigger than $\ell \p 1$, if the sets \textit{are not disjoint}. We have already shown  the existence of a pair with distance $2W \p \ell = \ell \p 2$ in the previous proof and are therefore done.	
\end{proof}

We move on to the reduction of the 2-party set disjointness problem to the diameter graph problem using the graph \smash{$\Gamma^{a,b}_{k,\ell,W}$} and its properties. Assume we are given a randomized \HYBRID model algorithm $\calA$ that computes the diameter $D$ of the graph. We assume that at the end of Algorithm $\calA$, at least one of the nodes of the graph must know the output of Algorithm $\calA$ (i.e., the diameter or an approximation of the diameter).

Assume that we are given a set disjointness instance, where Alice is given set $a$ and Bob is given set $b$. Then the two parties together simulate \smash{$\Gamma^{a,b}_{k,\ell,W}$} and run $\calA$ on the simulated  graph. 
During the simulation the players exchange messages that go between nodes simulated by opposing parties so they can continue the correct simulation. 

\begin{lem}
	\label{lem:player-simulation}
	Presume Alice knows \smash{$a\in \{0,1\}^{k^2}$} and Bob knows \smash{$b\in \{0,1\}^{k^2}$}. Let $\calA$ be a randomized \HYBRID model algorithm. Assume further that Alice and Bob have access to public randomness. Then the two parties can together \emph{correctly} simulate up to $\lfloor\ell/2\rfloor$ rounds of Algorithm $\calA$, such that (1) the only communication between Alice and Bob stems from messages via the global network and (2) each node is simulated by at least one party throughout the whole simulation of $\calA$.
\end{lem} 

\begin{proof}
	We organize the nodes of \smash{$\Gamma^{a,b}_{k,\ell,W}$} into columns. The first column consists of $V_1 \cup V_2 \cup \{ \hat v \}$. The last ($\ell^{th}$) column consists of $U_1 \cup U_2 \cup \{\hat u \}$. The columns in between are made up of the inner nodes on the paths that connect the first and the last column (c.f., Figure \ref{fig:diameter_graph}). Formally, column $i$ consists of all nodes that have hop-distance $i$ to their closest node in the first column.
	
	First some nomenclature. If we say that Alice (or Bob) simulates some round $r$ of a node $x$, we assume that Alice knows the correct state of node $x$ at the beginning of round $r$ (i.e., after finishing round $r-1$) and Alice can thus compute the messages sent by node $x$ in round $r$. This does \textit{not} necessarily require that Alice also learns all the messages $x$ receives in round $r$ (unless Alice also simulates round $r+1$ of $x$). This is due to the fact that $x$ can only react in round $r \p 1$ or later rounds to the messages $x$ receives in round $r$ anyway.
	
	Based on her knowledge of $a$, Alice is able to construct the subgraph of \smash{$\Gamma^{a,b}_{k,\ell,W}$} induced by the first $\ell\m 1$ columns. By symmetry, Bob can construct the subgraph induced by all columns except the first. Then the simulation works as follows. In the first round Alice simulates columns $1, \ldots , \ell\m 1$ and Bob simulates columns $2, \ldots , \ell$. With each subsequent simulation round the simulated node set shrinks by one column from the opposite side. Specifically, in round $i$, Alice only continues simulating columns $1, \ldots , \ell\m i$ and Bob $i\p 1, \ldots , \ell$. Since we only run $\calA$ for at most $\lfloor\ell/2\rfloor$ rounds, condition (2) is satisfied. 
	
	Note that the set of nodes simulated by Alice and Bob overlaps. In order to ensure a correct simulation of the protocol, a node $x$ that is simulated by both Alice and Bob has to be simulated consistently by the two parties. In particular, since $\calA$ can be a randomized algorithm, Alice and Bob have to use the same randomness in the simulation of node $x$. Note that they can consistently simulate the node because they have access to public randomness. In fact, if we think of the public randomness as part of the input given to the players, the whole simulation becomes deterministic.
	
	In order to make sure that the simulation of the protocol is correct, we need to guarantee that Alice and Bob always know the correct state of all the nodes they simulate. We prove that this is true by induction on the number of rounds. For the first round, clearly both Alice and Bob know the initial states of all the nodes they simulate and they can therefore perform the correct internal first round computations and compute the first round messages of all their nodes. For any round $r\geq 2$, we need to show that a correct simulation up to round $r-1$ implies that Alice and Bob can learn the state at the end of round $r-1$ (i.e., at the beginning of round $r$) for all the nodes they simulate in round $r$. 
	
	Consider some node $x$ that sends a message to node $y$ in round $r-1$. W.l.o.g., assume that Alice simulates node $x$ in round $r-1$ and Alice thus knows the message $x$ sends to $y$ in round $r-1$. If Bob also simulates node $x$ in round $r-1$, he also knows the message from $x$ to $y$ and Alice and Bob do not need to exchange any information about this message. Let us therefore assume that Bob does not simulate node $x$ in round $r-1$, but that Bob simulates node $y$ in round $r$. In this case, Bob needs to know the message from $x$ to $y$ in round $r-1$ to correctly simulate $y$ in round $r$ and Alice has to send the content of the message to Bob. 
	
	In order to show that Alice never needs to send information to Bob about a message that is sent over a local edge, we need to show that the following constellation can not occur. There can not be two neighbors $x$ and $y$ in $G$ such that $x$ \textit{is} simulated by Alice in round $r-1$, but $x$ is \textit{not} simulated by Bob in round $r-1$, and $y$ is simulated by Bob in round $r$. 
	
	In round $r-1$ Alice simulates the nodes in columns $1,\dots,\ell-(r-1)$ and Bob simulates the nodes in columns $r,\dots,\ell$. Node $x$ therefore needs to be in one of the columns $1,\dots,r-1$. However, in the next round $r$, Bob only simulates columns $r+1,\dots,\ell$ and therefore $y$ needs to be in one of those columns. Therefore, there cannot be an edge between $x$ and $y$ in $G$. For a correct simulation of $\calA$, Alice and Bob therefore only need to exchange information about messages sent over the global edges and thus also claim (1) of the lemma is proven.
\end{proof}

Now we have the required tools to prove Theorem \ref{thm:diameter_lower}. Specifically, we show that computing the diameter in the \HYBRID model takes \smash{$\Omega\big((n/\log^{2} n)^{1/3}\big)$} even for Monte Carlo algorithms with success probability $2/3$. Additionally, we want to show that the same is true for computing a $(2\m \eps)$-approximation for any $\eps >0$ by using only weights up to $n^{1/3}$.

\begin{proof}
	Assume $\calA$ is a (randomized) diameter algorithm that computes the diameter of a graph with probability at least $2/3$ and takes at most $\lfloor \frac{\ell}{2} \rfloor \m 1$ rounds for some $\ell \in [n]$ yet to be determined. Note that this implies that after correctly simulating $\calA$ for $\lfloor \frac{\ell}{2} \rfloor$ rounds every node knows the diameter of the graph.
	
	Consider the two-player set disjointness problem and assume that Alice and Bob are given two strings \smash{$a,b\in \set{0,1}^{k^2}$} as inputs. Then Lemma \ref{lem:player-simulation} tells us that Alice and Bob can together correctly simulate $\lfloor\frac{\ell}{2}\rfloor$ rounds Algorithm $\calA$ on graph \smash{$\Gamma^{a,b}_{k,\ell,W}$}. 
	This implies that Alice and Bob know the correct value of the diameter (in the unweighted case) or a $(2\m\eps)$-approximation of the diameter (in the weighted case) with probability at least $2/3$. By Lemmas \ref{lem:reduction_diam_weighted} and Lemma \ref{lem:reduction_diam_unweighted} this implies that they have solved the given set disjointness instance with probability $2/3$ (more on the weighted case further below).
	
	Due to Lemma \ref{lem:player-simulation}, the only information that Alice and Bob ever exchange during the simulation are the messages sent via local messages in the simulated \HYBRID network. But according to the lower bound for set disjointness~\cite{kalyanasundaram92,razborov92}, Alice and Bob must have exchanged $\Omega(k^2)$ bits during the whole process, since they solve the set disjointness problem with sufficiently good probability. 
	
	Specifically, this means that during the simulation $\Omega(k^2/\log n)$ messages must have been exchanged between nodes via the global network, since a message can hold at most $\bigO(\log n)$ bits. However, all nodes combined can send (or receive) at most $\bigO(n \log n)$ messages per round. Hence, it takes at least $\Omega(k^2/(n \log^2 n))$ rounds until $\calA$ is finished.
	
	Summing up, $\calA$ must take at least $\Omega\big(\min\big\{\lfloor \frac{\ell}{2} \rfloor \m 1, k^2/(n \log^2 n)\big\}\big)$ rounds. If $n$ is the total number of nodes of \smash{$\Gamma^{a,b}_{k,\ell,W}$}, then we know from the construction that $k \cdot \ell \in \Theta(n)$. This is fulfilled by choosing \smash{$\ell \in \Theta\big((n/\log^2 n)^{1/3}\big)$} and \smash{$k \in \Theta\big( (n \log n )^{2/3} \big)$}. Then the runtime of $\calA$ is \smash{$\Omega\big((n/\log^{2} n)^{1/3}\big)$}.
	
	In the weighted case, assume that the  approximation ratio of Algorithm $\calA$ is $\alpha$. In order to guarantee that Alice and Bob can distinguish disjoint from non-disjoint instances, $\alpha$ has to satisfy $\alpha(W \p 2\ell) < (2W \p \ell)$ (cf.\ Lemma \ref{lem:reduction_diam_weighted}). This means \smash{$\alpha < \frac{2W+\ell}{W+2\ell}$}. We have \smash{$\frac{2W+\ell}{W+2\ell} = 2 - o(1)$}, for any\smash{ $W \in \omega(\ell) \subseteq \Omega(n^{1/3})$}. Thus $\alpha = 2\m\eps$ fulfills the requirement for large enough $n$.
\end{proof}

\end{longversion}


	\iflong\newpage\fi
		
	\bibliographystyle{ACM-Reference-Format}
	\bibliography{ref/refcombined}	
		
	\iflong\newpage\fi
	
	\appendix
	
	\begin{longversion}

\section{Some Basic Probabilistic Concepts}
\label{apx:generalnotations}

\begin{lem}[Chernoff Bound]
	\label{lem:chernoffbound}
	We use the following forms of Chernoff bounds in our proofs:
		$$\mathbb{P}\big(X > (1 \!+\! \delta) \mu_H\big) \leq \exp\Big(\!-\!\frac{\delta\mu_H}{3}\Big),$$
	with $X = \sum_{i=1}^n X_i$ for i.i.d.\ random variables $X_i \in \{0,1\}$ and $\mathbb{E}(X) \leq \mu_H$ and $\delta \geq 1$. Similarly, for $\mathbb{E}(X) \geq \mu_L$ and $0 \leq \delta \leq 1$ we have
	$$\mathbb{P}\big(X < (1 \!-\! \delta) \mu_L\big) \leq \exp\Big(\!-\!\frac{\delta^2\mu_L}{2}\Big).$$
\end{lem}

\begin{remark}
	\label{rem:chernoffbound}
	Note that the first inequality even holds if we have $k$-wise independence among the random variables $X_i$ for $k \geq \lceil \mu_H \delta \rceil$ (c.f., \cite{Schmidt1995} Theorem 2, note that a substitution $\mu_H := (1 \!+\! \eps)\E(X)$ generalizes the result for any $\mu_H \geq \E(X)$).
\end{remark}

\begin{lem}[Union Bound]
	\label{lem:unionbound}
	Let $E_1, \ldots ,E_k$ be events, each taking place w.h.p. If $k \leq p(n)$ for a polynomial $p$ then $E \coloneqq \bigcap_{i=1}^{k} E_i$ also takes place w.h.p.
\end{lem}

\begin{proof}
	Let $d \coloneqq \deg(p)\!+\!1$. Then there is an $n_0 \geq 0$ such that $p(n) \leq n^d$ for all $n \geq n_0$. Let $n_1, \ldots , n_k \in \mathbb{N}$ such that for all $i \in \{1, \ldots, k\}$ we have $\mathbb{P}(\overline{E_i}) \leq \tfrac{1}{n^c}$ for some (yet unspecified) $c > 0$.
	With Boole's Inequality (union bound):
	\begin{align*}
	\mathbb{P}\big(\overline{E}\big) \!= \mathbb{P}\Big(\bigcup_{i=1}^{k} \overline{E_i} \Big) \leq \sum_{i=1}^{k} \mathbb{P}(\overline{E_i}) \leq \sum_{i=1}^{k} \!\frac{1}{n^c} \leq \frac{p(n)}{n^{c}} \leq \frac{1}{n^{c-d}}
	\end{align*}
	for all $n \geq n_0' \coloneqq \max(n_0, \ldots ,n_k)$. Let $c' > 0$ be arbitrary. We choose $c \geq c' \!\!+\! d$. Then $\mathbb{P}\big(\overline{E}\big) \leq \frac{1}{n^{c'}}$ for all $n \geq n_0'$.
\end{proof}

\begin{remark}
	If a constant number of events is involved we use the above lemma without explicitly mentioning it. It is possible to use the lemma in a nested fashion as long as the number of applications is polynomial in $n$.
\end{remark}

\end{longversion}

\section{Dissemination Protocols In HYBRID}

In this section we give a quick overview about the known protocols to disseminate information in \HYBRID networks that we will use in this paper. The first one is the so-called token dissemination protocol of \cite{Augustine2020}, which solves the problem to broadcast $k$ tokens from arbitrary sources, whereas each source has at most $\ell$ tokens.

\begin{lem}(Theorem 2.1 in \cite{Augustine2020})
	\label{lem:tokenDissemination}
	Assume there are $k$ tokens of size $\bigO(\log n)$ bits in the network and each node has at most $\ell$ tokens. There is an algorithm that solves the token dissemination problem in the \HYBRID model on connected graphs in \smash{$\tilO\big(\!\sqrt{k}  \!+\! \ell\big)$} rounds, w.h.p.
\end{lem}

Another protocol that we use frequently, solves the following aggregation problem. Given a subset of nodes $A \subseteq V$ that have some information $x_v, v \in A$, all nodes must learn the result $f\big(x_v \mid v\in A\big)$, where $f$ is an \textit{aggregate distributive} function.\lng{\footnote{An \textit{aggregate} function $f$ maps a multiset $S$ of inputs to some $f(S)$. For an \textit{aggregate distributive} function $f$ there is an aggregate function $g$ such that for any partition of \smash{$S = S_1 \oplus , \dots , \oplus S_k$} it is \smash{$f(S) = g(f(S_1), \dots , f(S_k))$}}. An example is computing a sum or the maximum value.} The algorithm by \cite{augustine19_NCC} solves this problem, relying only on the global network.

\begin{lem}
	\label{lem:aggregationProtocol}(Theorem 2.1 in \cite{augustine19_NCC})
	The aggregation problem can be solved in in $O(\log n)$ in the \NCC model.
\end{lem}

\section{Techniques for Skeleton graphs}
\label{apx:skeleton-graphs}

By sampling nodes from $G$ with prob.\ $\frac{1}{x}$ there will be a sampled node on some shortest $u$-$v$-path $P$ at least every $\tilO(x)$ hops for any $u,v \in V$ w.h.p. We formalize that in the following lemma. This was first observed by \cite{Ullman1991}. \shrt{Proofs are provided in the full article \cite{KS20}.}

\begin{lem}[From \cite{Augustine2020}]
	\label{lem:long-paths-skeleton}
	Let $M\subseteq V$ be a subset of nodes of $G=(V,E)$ by sampling each node with probability at least $\frac{1}{x}$. Then there is a constant $\xi \!>\! 0$, such that for any $u,v \!\in\! V$ with $hop(u,v) \!\geq\! \xi x \ln n$, there is at least one shortest path $P$ from $u$ to $v$, such that any sub-path $Q$ of $P$ with at least $\xi x \ln n$ nodes contains a node in $M$ w.h.p.
\end{lem}

\begin{longversion}
	
	\begin{proof}
		Let $u,v \in V$ with $hop(u,v) \!\geq\! \xi x \ln n$. Fix a shortest $u$-$v$-path $P_{u,v}$ and let $Q$ be a sub-path of $P_{u,v}$ with at least $\xi x \ln n$ nodes. Let $X_{u,v}$ be the random number of marked nodes on $Q$. Then we have \smash{$\mathbb{E}(X_{u,v}) \geq \frac{|Q|}{x} \geq {\xi \ln n}$}. Let $c > 0$ be arbitrary. We use a Chernoff bound:
		$$ \mathbb{P}\Big( X_{u,v} < \frac{\xi \ln n}{2}\Big) \leq \exp\Big(\!\!-\! \frac{\xi \ln n}{8}\Big) \stackrel{\xi \geq 8c}{\leq} \frac{1}{n^c}.$$
		Thus we have $X_{u,v} \geq 1$ w.h.p.\ for constant \smash{$\xi \geq \max(8c, 2 / \ln n)$}. Therefore the claim holds w.h.p.\ for the pair $u,v$. We claim that w.h.p.\ the event $X_{u,v} \geq 1$ occurs for all pairs $u,v \in V$ and for all sub-paths $Q$ of $P_{u,v}$ longer than $\xi x \ln n$ hops, for at least one shortest path $P_{u,v}$ from $u$ to $v$. There are at most \smash{$n^2$} many pairs $u,v \in V$. Moreover we can select at most $n$ sub-paths $Q$ of $P$ that do not fully contain any other selected sub-path. Hence the claim follows with the union bound given in Lemma~\ref{lem:unionbound}.
	\end{proof}
	
\end{longversion}

The above lemma implies additional properties of skeletons that we use (sometimes implicitly) in this paper.

\begin{lem}
	\label{lem:long-shortest-paths-skeleton}
	Let $\calS = (V_\calS, E_\calS)$ be a skeleton graph of a connected graph $G$ with $n$ nodes by sampling each node of $G$ to $V_\calS$ with prob.\ at least $\frac{1}{x}$. The edges of $\calS$ are $E_\calS \!=\! \{ \{u,v\} \!\mid\! u,v\!\in\!V_\calS, \text{hop}(u,v) \!\leq\! h\}$ (whereas $h := \xi x \ln n$ with parameter $\xi$ as in Lemma \ref{lem:long-paths-skeleton}) and edge weights $d_h(u,v)$ for $\{u,v\} \in E_\calS$. Then $\calS$ is connected w.h.p. Furthermore, for any $u,v\in \calS$ we have $d_G(u,v) = d_\calS(u,v)$ w.h.p.
\end{lem}

\begin{longversion}
	
	\begin{proof}
		Let $u,v \in \calS$ and let $P$ be a shortest $u$-$v$-path in $G$. By Lemma \ref{lem:long-paths-skeleton} there is a skeleton node at least every $h$ hops on $P$ w.h.p. Hence, consecutive skeleton nodes must be connected by edges in $E_\calS$ by definition. These form a $u$-$v$-path $Q$ in $\calS$. Therefore $u,v$ are connected in $G$.
		
		Let $u = v_1, \ldots ,v_k = v$ be the set of consecutive skeleton nodes on $P$. Then $hop(v_i, v_{i+1}) \leq h$ and $\{v_i, v_{i+1}\} \in E_\calS$ due to Lemma \ref{lem:long-paths-skeleton}. Since $P$ is a shortest path, all its sub-paths must be shortest paths as well. Hence we have $d_h(v_i, v_{i+1}) = d(v_i, v_{i+1})$. Therefore $d_\calS(u,v) \leq \sum_{i=1}^{k-1} d(v_i, v_{i+1}) =d_G(u,v)$. Since edges of $\calS$ have at least the weight of the path in $G$ they represent, we also have $d_\calS(u,v) \geq d_G(u,v)$.
	\end{proof}
	
\end{longversion}

\begin{longversion}

\section{$k$-Wise Independent Hash Functions}
\label{apx:Pseudorandom}

We will use families of $k$-wise independent hash functions, which are defined as follows.

\begin{defn}
	\label{def:hashfunctions}	
	For finite sets $A,B$, let $\mathcal H := \{h : A \to B\}$ be a family of hash functions. Then $\mathcal H$ is called $k$-wise independent if for a random function $h \in \mathcal H$ and for any $k$ distinct keys $a_1, \ldots, a_k \in A$ we have that $h(a_1), \ldots, h(a_k) \in B$ are independent and uniformly distributed random variables in $B$.
\end{defn}

From literature we know that such a family of hash functions exists in the following form (c.f., \cite{Vadhan2012}).

\begin{lem}
	\label{lem:hashfunctions}
	For $A := \{0,1\}^a$ and $B := \{0,1\}^b$, there is a family of $k$-wise independent hash functions $\mathcal H := \{h : A \to B\}$ such that selecting a function from $\mathcal H$ requires $k \!\cdot\! \max(a,b)$ random bits and computing $h(x)$ for any $x \in A$ can be done in $poly(a,b,k)$ time.
\end{lem}

%

\begin{remark}
	We can use a random member $h \in \mathcal H$ of a family as described in Definition \ref{def:hashfunctions} to limit the number of messages any node receives in a given round. If all nodes send at most $\bigO(\log n)$ messages to targets that are determined using that hash function $h$ with a distinct key for each message, then any node receives at most $\bigO(\log n)$ messages per round w.h.p. The details are given by the following lemma.
\end{remark}

\begin{lem}
	\label{lem:receiveBound}
	Presume some algorithm takes at most $p(n)$ rounds for some polynomial $p$. Moreover, presume that each round, every node sends at most $\sigma \in \Theta(\log n)$ messages via global edges to $k$-wise independent, uniformly random targets in $V$. There is a $k \in \Theta(\log n) $ and a $\rho \in \Theta(\log n)$ such that for sufficiently large $n$, in every round, every node in $V$ receives at most $\rho$ messages per round w.h.p.
\end{lem}

\begin{proof}
	Let $X_{v,r}$ be the (random) number of messages $v$ receives in round $r$. Node $v$ is targeted by some token for some node from $V$ with probability  $\tfrac{1}{n}$ (for simplicity we assume that $v$ can send the token to itself, in reality it can just keep it). Since each node sends at most $\sigma$ tokens per round we have $\mathbb{E}\big(X_{v,r}\big) \leq \sigma \cdot n \cdot \tfrac{1}{n} = \sigma$. By definition we have $\sigma \geq \xi \ln n$ for some constant $\xi$ and large enough $n$. Let $c>0$ be arbitrary.  We choose \smash{$\rho \geq (1 \!+\! \frac{3c}{\xi}) \sigma$} and $k \geq \lceil \frac{3c}{\xi} \sigma \rceil$. We employ the Chernoff bound given in Lemma \ref{lem:chernoffbound}, which holds for $k$-wise independent random variables as outlined in Remark \ref{rem:chernoffbound}.
	$$\mathbb{P}\Big(X_{v,r} \!>\! \rho \Big) \leq \mathbb{P}\Big(X_{v,r} \!>\! (1 \!+\! \tfrac{3c}{\xi})\sigma \Big) \leq \exp\Big(\!\!-\! \frac{3\xi c \ln n}{3\xi}\Big) =  \frac{1}{n^{c}}.$$
	In accordance with the union bound given in Lemma~\ref{lem:unionbound} the event \smash{\hspace*{0.1mm} $\bigcap_{\substack{\!\!\!v \in V\\ \hspace*{-0.5mm}\!\! r \leq p(n)}} \hspace*{-1mm} \Big( X_{v,r} \!\leq\! \rho \Big)$} takes place w.h.p.
\end{proof}

\end{longversion}

\end{document}